\documentclass[a4paper, 10 pt, conference]{ieeeconf} 

\IEEEoverridecommandlockouts                              
\AtBeginDocument{%
	\providecommand\BibTeX{{%
			\normalfont B\kern-0.5em{\scshape i\kern-0.25em b}\kern-0.8em\TeX}}}

\let\labelindent\relax
\usepackage{amsmath}
\usepackage{amsfonts}
\usepackage{enumerate}
\usepackage{amsthm}
\usepackage{paralist}
\usepackage[linesnumbered, ruled]{algorithm2e}
\usepackage{algpseudocode}
\usepackage{comment}
\usepackage{mathtools}
\usepackage{braket}
\usepackage{multirow}
\usepackage{caption}
\usepackage{subcaption}
\usepackage{siunitx}
\usepackage{xspace}
\usepackage{pifont}
\usepackage{booktabs}
\usepackage{url}
\newtheorem{theorem}{Theorem}
\newtheorem{problem}{Problem}
\newtheorem{definition}{Definition}[section]
\newtheorem{corollary}{Corollary}
\usepackage{wrapfig}
\usepackage{mathrsfs}
\usepackage{listings}
\usepackage{enumitem}
\usepackage{cleveref}

\newcommand{\BW}[1]{{\color{red} \textbf{Ben:} #1}}

\newcommand{\Sol}{\varphi}
\newcommand{\Solnom}{\bar{\varphi}}
\newcommand{\Xh}{\widehat{X}}
\newcommand{\xh}{\hat{x}}
\newcommand{\Uh}{\widehat{U}}
\newcommand{\Sigmah}{\widehat{\Sigma}}
\newcommand{\Ch}{\widehat{C}}
\newcommand{\uh}{\hat{u}}
\newcommand{\fh}{\hat{f}}

\newcommand{\ball}{\Omega}
\def\reals{\mathbb{R}}
\def\nats{\mathbb {N}}

\def\prob{\mathbb{P}}
\def\borel{\mathscr{B}}
\def\D{\mathcal D}

\def\init{\mathsf{in}}

\newtheorem{proposition}{Proposition}
\newtheorem{remark}{Remark}
\newtheorem{assumption}{Assumption}

\usepackage[many]{tcolorbox}
\usetikzlibrary{calc}
\tcbuselibrary{skins}
\newtcolorbox{resp}[1][]{%
	enhanced jigsaw,%
	colback=gray!5!white,%
	colframe=gray!80!black,%
	size=small,%
	boxrule=1pt,%
	halign title=flush center,%
	coltitle=black,%
	breakable,%
	drop shadow=black!50!white,%
	attach boxed title to top left={xshift=1cm,yshift=-\tcboxedtitleheight/2,yshifttext=-\tcboxedtitleheight/2},%
	minipage boxed title=3cm,%
	boxed title style={%
		colback=white,%
		size=fbox,%
		boxrule=1pt,%
		boxsep=2pt,%
		underlay={%
			\coordinate (dotA) at ($(interior.west) + (-0.5pt,0)$);
			\coordinate (dotB) at ($(interior.east) + (0.5pt,0)$);
			\begin{scope}[gray!80!black]
				\fill (dotA) circle (2pt);
				\fill (dotB) circle (2pt);
			\end{scope}
		}%
	},%
	#1%
}

\begin{document}

\title{Data-Driven Abstraction-Based Control Synthesis}

\author{Milad Kazemi$^{1}$, Rupak Majumdar$^{2}$, Mahmoud Salamati$^{2}$, Sadegh Soudjani$^{1}$, Ben Wooding$^{1}$
\thanks{*This work is supported by the the EPSRC New Investigator Award CodeCPS (EP/V043676/1) and by the DFG project 389792660 TRR 248--CPEC.}
\thanks{$^{1}$ School of Computing, Newcastle University, United Kingdom
        }%
\thanks{$^{2}$ MPI--SWS, Kaiserslautern, Germany
        }%
}







\maketitle
\begin{abstract}
This paper studies formal synthesis of controllers for continuous-space systems with unknown dynamics to satisfy requirements expressed as linear temporal logic formulas. Formal abstraction-based synthesis schemes rely on a precise mathematical model of the system to build a finite abstract model, which is then used to design a controller. The abstraction-based schemes are not applicable when the dynamics of the system are unknown. We propose a data-driven approach that computes the growth bound of the system using a finite number of trajectories. The growth bound together with the sampled trajectories are then used to construct the abstraction and synthesise a controller. 

Our approach casts the computation of the growth bound as a robust convex optimisation program (RCP). Since the unknown dynamics appear in the optimisation, we formulate a scenario convex program (SCP) corresponding to the RCP using a finite number of sampled trajectories.
We establish a sample complexity result that gives a lower bound for the number of sampled trajectories to guarantee the correctness of the growth bound computed from the SCP with a given confidence. 
We also provide a sample complexity result for the satisfaction of the specification on the system in closed loop with the designed controller for a given confidence. Our results are founded on estimating a bound on the Lipschitz constant of the system and provide guarantees on satisfaction of both finite and infinite-horizon specifications.
We show that our data-driven approach can be readily used as a model-free abstraction refinement scheme by modifying the formulation of the growth bound and providing similar sample complexity results. The performance of our approach is shown on three case studies.
\end{abstract}

\section{Introduction}
\label{sec:intro}

One of the major objectives in the design of safety-critical systems is to ensure their safe operation while satisfying high-level requirements. Examples of safety-critical systems include power grids, autonomous vehicles, traffic control, and battery-powered medical devices.  Automatic design of controllers for such systems that can fulfil the given requirements have received significant attention recently. These systems can be represented as control systems with continuous state spaces. Within these continuous spaces, it is challenging to leverage automated control synthesis methods that provide satisfaction guarantees for high-level specifications, such as those expressed in Linear Temporal Logic \cite{katoen2008book, belta2017formal, Tabuada2009book, girard2007approximation}.

A common approach to tackle the continuous nature of the state space is to use abstraction-based controller design (ABCD) schemes \cite{Tabuada2009book,belta2017formal,rupak2020ouputfeedback,stanly2020}. The first step in the ABCD scheme is to compute a finite abstraction by discretising the state and action spaces. Finite abstractions are connected to the original system via an appropriate behavioural relation such as feedback refinement relations or alternating bisimulation relations \cite{reissig2016feedback,Tabuada2009book}. Under such behavioural relations, trajectories of the abstraction are related to the ones of the original system. Therefore, a controller designed for the simpler finite abstract system can be refined to a controller for the original system. The controller designed by the ABCD scheme is described as being formal due to the guarantees on satisfaction of the specification by the original system in closed loop with the designed controller.

ABCD schemes generally rely on a precise mathematical model of the system. This stems from the fact that establishing a behavioural relation between the original system and its finite abstraction uses reachability analysis over the dynamics of the original system that require knowledge of the dynamical equations. Although such equations can in principle by derived for instance by using physics laws, the real-world control systems are a mixture of differential equations, block diagrams, and lookup tables. Therefore, extracting a clean analytical model for systems of practical interest could be infeasible. A promising approach to tackle this issue is to develop data-driven control synthesis schemes with appropriate formal (probabilistic) guarantees.
 
 The main contribution of this paper is to provide a data-driven approach for formal synthesis of controllers to satisfy temporal specifications. We focus on continuous-time nonlinear dynamical systems whose dynamics are unknown but sampled trajectories are available. Our approach constructs a \emph{finite} abstract model of the system using only a finite number of sampled trajectories and the growth bound of the system. We formulate the computation of the growth bound as a robust convex program (RCP) that has infinite uncountable number of constraints. We then approximate the solution of the RCP with a scenario convex program (SCP) that has a finite number of constraints and can be solved using only a finite set of sampled trajectories. We establish a sample complexity result that gives a lower bound for the required number of trajectories to guarantee the correctness of the growth bound \emph{over the whole state space} with a given confidence.
We also provide a sample complexity result for the satisfaction of the specification on the system in closed loop with the designed controller for a given confidence. Our result requires estimating a bound on the Lipschitz constant of the system with respect to the initial state, that we obtain using extreme value theory.
As our last contribution, we show that our approach can be extended to a model-free abstraction refinement scheme by modifying the formulation of the growth bound and providing similar sample complexity results. We demonstrate the performance of our approach on three case studies.
 
 
 \smallskip
 
 The remainder of this paper is organised as follows.
 After discussing the related work, Section~\ref{sec: prelims} covers preliminaries on dynamical systems and finite abstractions, and provides the problem statement.
 In Section~\ref{sec:RCPs}, we present the assumptions and theoretical results needed for connecting RCPs and their corresponding SCPs. 
 In Section~\ref{sec:main_results}, we present our approach on data-driven computation of the growth bound and the abstraction, and prove our sample complexity result. Estimation of the Lipschitz constant of the system for computing the number of samples is also discussed in this section.
Section~\ref{sec:refinement} discusses the extension of our approach to a data-driven abstraction refinement scheme.
Several numerical examples are provided in Section~\ref{sec:experiments} that support the theoretical findings of our paper. Finally, Section~\ref{sec:conclusion} contains concluding remarks and future research directions.

 
 \vspace{0.1cm}
\noindent\textbf{Related Work.}
There is an extensive body of literature on \emph{model-based} formal synthesis for both deterministic and probabilistic systems. We refer the reader to the books \cite{katoen2008book,Tabuada2009book,belta2017formal} and seminal papers \cite{girard2007approximation, abate2008probabilistic}.
\emph{Data-driven} approaches for analysis, verification, and synthesis of systems have received significant attention recently to improve efficiency and scalability of model-based approaches, and to study problems in which a model of the system is either not available or costly and time-consuming to construct.
 
Given a prior inaccurate knowledge about the model of the system, a research line is to use data for refining the model and then synthesise a controller. Such approaches assume a class of models and improve the estimation of the uncertainty within the model class. These approaches range from using Gaussian processes \cite{mitsioni2021safe, bajcsy2019efficient}, differential inclusions \cite{djeumou2020fly}, rapidly-exploring random graphs \cite{grover2021semantic}, piecewise affine models \cite{Belta:2018}, and model-based reinforcement learning algorithms \cite{Belta:2021}.
A data-driven framework is proposed in \cite{Fan:2017} for verifying properties of hybrid systems when the continuous dynamics are unknown but the discrete transitions are known.

Data-driven model-free approaches compute the solution of the synthesis problem directly from data without constructing a model.
In \cite{hsu2021safety}, authors provide a reach-avoid Q-learning algorithm with convergence guarantees for an arbitrarily tight conservative approximation of the reach-avoid set. 
The paper \cite{wang2020falsification} proposes a falsification-based adversarial reinforcement learning algorithm for metric temporal logic specifications.
Satisfying signal temporal logic specifications is studied in \cite{verdier2020formal} using counterexample-guided inductive synthesis on nonlinear systems, and using model-free reinforcement learning in \cite{kalagarla2021model} for Markov decision processes.
A learning framework for synthesis of control-affine systems in provided in \cite{sun2020learning}.
The authors of \cite{Watanab21} study learning from demonstration while preventing the violation of safety under the learned policy.


The research on data-driven constructions of abstract models is very limited. Legat et al.~\cite{legat2021abstraction} provide an abstraction-based controller synthesis approach for hybrid systems by computing Lyapunov functions and Bellman-like Q-functions, and using a branch and bound algorithm to solve the optimal control problem. Makdesi et al.~\cite{Girard:2021} studied unknown monotone dynamical systems and sampled a set of trajectories generated by the system to find a minimal map overapproximating the dynamics of any system that produces these transitions. Consequently, they calculate an abstraction of the system related to this map and prove that an alternating bisimulation relation exists between them. In contrast, our approach is not restricted to monotone systems and is applicable to any nonlinear dynamical system. 

\smallskip

The closest work to our problem formulation is the work by Devonport et al.~\cite{Arcak:2021}, where a data-driven abstraction technique is provided for satisfying finite-horizon specifications. Our results are more general than the work \cite{Arcak:2021} in two main aspects. First, our constructed abstraction can be used for synthesising a controller against any linear temporal logic specification. Our sample complexity result is independent of the horizon of the specification and does not limit using the approach on finite-horizon specifications. Second, the guarantee provided in \cite{Arcak:2021} is based on a Probably Approximately Correct (PAC) approach. It means that the constructed abstraction is always wrong on a small subset of the state space whose size can be made smaller at the cost of high computational efforts. Our formulated guarantee ensures that the abstraction is valid on the entire state space with high confidence. The confidence is interpreted from the frequentist view of probability: if we run our algorithm multiple times, we always get a correct abstraction except a small number of times reflected in the confidence value.



\smallskip

In our approach, we formulate the synthesis problem as a robust convex program and approximate it with a scenario program.
Such approximations have been studied for the past two decades.
Calafiore and Campi~\cite{calafiore2006scenario} provide an approximately feasible solution for the associated chance constrained program by solving a scenario program, and give a sample complexity result.
Relaxing the convexity assumption is studied in \cite{soudjani2018concentration} by assuming additional properties of the underlying probability distributions.
%
We will use the results by Esfahani et al.~\cite{esfahani2014performance}, where the optimality of the robust program is connected directly to the scenario program. These results are also used recently in the papers \cite{lavaei2021formal,salamati2021data} for performing data-driven verification and synthesis. Inspired by the works \cite{wood1996estimation,weng2018evaluating}, we will use extreme value theory to estimate the Lipschitz constant needed for the sample complexity results.


\section{Preliminaries and Problem Statement}\label{sec: prelims}
\subsection{Preliminaries}
\noindent\textbf{Notation.}\
We denote the set of natural, real, positive real, and non-negative real numbers by $\nats$, $\reals$, $\reals_{>0}$, and $\reals_{\ge 0}$, respectively.
The set of natural numbers including zero is denoted by $\nats_{\ge 0}$.
We use superscript $n>0$ with these sets to denote the Cartesian product of $n$ copies of these sets.
The power set of a set $A$ is denoted by $2^A$ and includes all the subsets of $A$.
For any $x,y\in\reals^n$ with $x=(x_1,\ldots, x_n)$ and $y=(y_1,\ldots, y_n)$,
and a relational symbol $\triangleright\in \set{\leq, <, = , >, \geq}$, we write $x\triangleright y$ if $x_i\triangleright y_i$ for every $i\in\set{1,2,\ldots,n}$. A matrix $M\in\reals^{n\times n}$ is said to be non-negative if all of its entries are non-negative.
 We use the operators $|\cdot|$
 and $\|\cdot\|$ to denote the element-wise absolute value
 and the infinity norm, respectively. We use the notation $\ball_\varepsilon(c)\coloneqq \set{x \in \reals^n \mid  \|x-c\|\leq \varepsilon}$ to denote the ball with respect to infinity norm centred at $c\in \reals^n$ with radius $\varepsilon\in \reals_{>0}^{n}$. 
We consider a probability space $(\Omega,\mathcal F_{\Omega},\prob_{\Omega})$, where $\Omega$ is the sample space, $\mathcal F_{\Omega}$ is a sigma-algebra on $\Omega$ comprising its subsets as events, and $\prob_{\Omega}$ is a probability measure that assigns probabilities to events.

\smallskip
\noindent\textbf{Control Systems.}\
A continuous-time control system is a tuple $\Sigma = (X, x_\init, U, W, f)$, where
$X\subset \reals^n$ is the state space,
$x_\init\in X$ is the initial state,
$U\subset\reals^m$ is the input space, and
$W\subset\reals^n$ is the disturbance space which is assumed to be a compact set containing the origin.
The vector field $f: X \times U \rightarrow X$ is such that $f(\cdot, u)$ is locally Lipschitz for all $u\in U$.
The evolution of the state of $\Sigma$ is characterised by the differential equation
\begin{equation}
\label{eq:ODE}
\dot x(t)=f(x(t),u(t))+w(t),
\end{equation}
where $w(t)\in W$ represents the additive disturbance.

We consider the class of input and disturbance signals $u:\reals_{\geq0}\rightarrow U$ and $w:\reals_{\geq0}\rightarrow W$ to be piecewise constant with respect to a \emph{sampling time} $\tau>0$, i.e.,  $u(t)=u(k\tau)$ and $w(t)=w(k\tau)$ for every $k\tau\leq t< (k+1)\tau$ and $k\in\nats_{\ge 0}$.
Given a sampling time $\tau>0$, an initial state $x_0\in X$, a constant input $u\in U$, and a constant disturbance $w\in W$, define the \emph{continuous-time trajectory} $\zeta_{x_0, u, w}$
of the system on the time interval $[0, \tau]$ as an absolutely continuous function $\zeta_{x_0,u,w}: [0, \tau] \rightarrow X$ such that $\zeta_{x_0,u,w}(0) = x_0$, and
$\zeta_{x_0,u,w}$ satisfies the differential equation $\dot{\zeta}_{x_0,u,w}(t) = f(\zeta_{x_0,u,w}(t), u) + w$ for almost all $t\in [0,\tau]$.
The solution of \eqref{eq:ODE} from $x_0$ for the constant control input $u$ with $w(t)=0$ for all $t\ge 0$ is called the \emph{nominal trajectory} of the system.
For a fixed $\tau$, we define the operators
\begin{align*}
	\Sol(x, u, w)&:=\zeta_{x,u,w}(\tau) \,\text{ and }\\
	\Phi(x, u)&:=\set{\Sol(x, u,w)\mid w\in W}
\end{align*} 
respectively for the trajectory at time $\tau$ and the set of such trajectories starting from $x$.

In this paper, we consider control systems $\Sigma = (X, x_\init, U, W, f)$ whose vector field $f$ is not known, but we can observe their \emph{time-sampled trajectories}. A sequence $x_0,x_1,x_2,\ldots$ is a time-sampled trajectory of $\Sigma$ if for each $i\geq 0$, we have $x_{i+1} \in \Phi(x_i, u_i)$ for some $u_i \in U$.

\smallskip
\noindent\textbf{Finite-state Abstraction of Control Systems.}\
Let $\Sigma = (X,x_\init, U, W, f)$ 
 be a control system with a sampling time $\tau>0$.
We consider abstract models constructed by using uniformly sized rectangular partitioning of $X$ and $U$.
We select representative points from these partition sets to obtain $\Xh$ and $\Uh$.
We assume that the radius of these partition sets are provided as vectors $\eta_x\in \mathbb{R}^n_{>0}$ and $\eta_u\in\reals^m_{>0}$, respectively. Parameters $\eta_x,\eta_u$ are inputs to the abstraction procedure.
A \emph{finite-state abstraction} of $\Sigma$ is characterised by the tuple $\Sigmah = (\Xh,\Uh,\fh)$, where 
$\Xh$ is the set of representative points from a finite partition of $X$,
$\Uh$ is the set of representative points from a finite partition of $U$,
and $\fh: \Xh\times \Uh\rightarrow 2^{\Xh}$ is a set-valued map.
For any $\xh\in\Xh$ and $\uh\in\Uh$,
$\xh'\in\fh(\xh,\uh)$ if there is a pair of states $x\in \ball_{\eta_x}(\xh)$ and $x'\in \ball_{\eta_x}(\xh')$ such that $x'\in \Phi(x,\uh)$.
Note that, the larger $\eta_x$ is (where comparison is made dimension-wise), the smaller is the cardinality of $\Xh$ resulting in a coarser abstraction.
On the other hand, the smaller $\eta_x$ is, the more precise the abstraction $\widehat{\Sigma}$ will be, increasing the chance of a successful controller synthesis (see, e.g., \cite{Tabuada2009book} for more details on this construction). 

\smallskip\noindent
\textbf{Feedback Controller.} 
A \emph{feedback controller} for $\Sigmah$
is a function $\Ch\colon \Xh\to \Uh$.
We denote by $\Ch \parallel \Sigmah$ the feedback composition of $\Sigmah$ and $\Ch$.
The set of trajectories of the closed-loop system $\Ch\parallel\Sigmah$ consists of all finite trajectories
$\xh_0, \xh_1,\xh_2,\ldots$ such that for all $i \in \nats_{\ge 0}$, we have $\xh_{i+1} \in \fh(\xh_i,\Ch(\xh_i))$. 

We can relate a finite abstraction $\Sigmah$ to $\Sigma$ for control synthesis purposes. Simulation relations or feedback refinement relations \cite{Tabuada2009book,reissig2016feedback} established between $\Sigma$ and $\Sigmah$ enable us to refine a controller $\Ch$ designed for $\Sigmah$ to a controller $C$ for $\Sigma$. In its general form, such a refined controller $C$ maps the current states $x\in \ball_{\eta_x}(\xh)$ into an input $u =\Ch(\xh)$ for $\Sigma$. 
 The purpose of designing $\Ch$ is that the closed-loop system $C\parallel\Sigma$ satisfies the given objective.
%
Our synthesis objective is expressed as Linear Temporal Logic (LTL) specifications. We refer to \cite{katoen2008book} and references therein for detailed syntax and semantics of LTL. For the details of the controller synthesis and tool implementation using abstract models we refer to \cite{reissig2016feedback} and \cite{rungger2016scots}, respectively. 


\subsection{Problem Statement}
\label{subsec:problem_statement}
We study abstraction-based control design (ABCD) for systems with \emph{unknown} dynamics using available data from the system such that a given specification is satisfied with high confidence on the closed-loop system.
\begin{assumption}
The vector field $f$ of the control system  $\Sigma=(X,x_\init,U,W,f)$ in unknown, but sampled trajectories of the system can be obtained in the form of $\mathcal S_N:= \{(x_k,u_k,x'_k)\,|\, x'_k \in \Phi(x_k,u_k), \,k=1,2,\ldots,N\}$.
\end{assumption}
\begin{resp}
	\begin{problem}[Data-driven ABCD]
		\label{prob:synthesis}
		\noindent\textbf{Inputs:} Control system $ \Sigma=(X,x_\init,U,W,f)$ with unknown vector field $f$, specification $\Psi$, sampled trajectories $\mathcal S_N$, and confidence parameter $\beta\in(0,1)$.
		
		\smallskip
		
		
		\noindent\textbf{Outputs:} Abstract model $\Sigmah$, abstract controller $\Ch$, and refined controller $C$ for $\Sigma$, such that $C\parallel \Sigma$ satisfies $\Psi$ with confidence $(1-\beta)$. 
	\end{problem}
\end{resp}

The first step of the ABCD is to compute a finite abstraction $\Sigmah$ for $\Sigma$. Once such an abstraction is computed, synthesis of the controller $\Ch$ and refining it to $C$ follow the model-based ABCD scheme.
Therefore, the main challenge is to provide a data-driven computation of the abstraction $\Sigmah$ that is a true overapproximation of $\Sigma$ with confidence $(1-\beta)$.
\begin{resp}
	\begin{problem}[Data-driven Abstraction]
		\label{prob:abstraction}
		\noindent\textbf{Inputs:} 
		Control system $ \Sigma=(X,x_\init,U,W,f)$ with unknown vector field $f$, sampled trajectories $\mathcal S_N$, discretisation parameters $\eta_x$ and $\eta_u$, and confidence parameter $\beta\in(0,1)$.
		
		\smallskip
		
		\noindent\textbf{Outputs:} Finite model $\Sigmah$ that is an abstraction of $\Sigma$ with confidence $(1-\beta)$.
	\end{problem}
\end{resp} 
In this paper, we tackle Problem~\ref{prob:abstraction} by showing how to construct $\Sigmah$ from sampled trajectories $\mathcal S_N$, and provide a lower bound on the data size $N$ in order to ensure correctness of the abstraction with confidence $(1-\beta)$. The required theoretical tools are presented in the next section.


\section{Robust Convex Programs}
\label{sec:RCPs}
In this section, we describe robust convex programs (RCPs) and data-driven approximation of their solution. In Sections~\ref{sec:main_results} and~\ref{sec:refinement}, we show how such an approximation can be used for solving the data-driven abstraction in Problem~\ref{prob:abstraction}. 

Let $T\subset\reals^{q}$ 
 be a compact convex set for some $q\in\nats$ and $c\in\reals^{q}$ be a constant vector. Let $(\D,\borel,\prob)$ be the probability space of the \textit{uncertainty} and $g\colon T\times\D\rightarrow\reals$ be a measurable function, which is
 	convex in the first argument for each $d \in\mathcal D$, and bounded in
 	the second argument for each $\theta\in T$. 
  The robust convex program (RCP) is defined as
\begin{equation}
\label{eq:RCP_def}
\text{RCP: }
\begin{cases}
	\,\,\min_{\theta} c^\top\theta\\
	\,\,s.t.\,\, \theta\in T \text{ and }\,\, g(\theta,d)\leq 0\quad \forall d\in\D.
	\end{cases}
\end{equation}
Computationally tractable approximations of the optimal solution of the RCP~\eqref{eq:RCP_def} can be obtained using \emph{scenario convex programs} (SCPs) that only require gathering finitely many samples from the uncertainty space \cite{Esfahani2015scp}. 
Let $(d_i)_{i=1}^N$ be $N$ independent and identically distributed (i.i.d.) samples drawn according to the probability measure $\prob$. 
The SCP corresponding to the RCP~\eqref{eq:RCP_def} strengthened with $\gamma\ge 0$ is defined as
\begin{equation}
	\label{eq:SCP_def}
	SCP_{\gamma}:
	\begin{cases}
		\,\,\min_{\theta} c^\top\theta\\
		\,\,s.t.\,\, \theta\in T, \text{ and }\,\, g(\theta,d_i)+\gamma\leq 0 \,\, \forall i\in\{1,2,\dots,N\}.
	\end{cases}
\end{equation}
We denote the optimal solution of RCP~\eqref{eq:RCP_def} as $\theta^\ast_{RCP}$ and the optimal solution of $SCP_{\gamma}$~\eqref{eq:SCP_def} as $\theta^\ast_{SCP}$. Note that $\theta^\ast_{RCP}$ is a single deterministic quantity but $\theta^\ast_{SCP}$ is a random quantity that depends on the i.i.d. samples $(d_i)_{i=1}^N$ drawn according to $\prob$.
The RCP~\eqref{eq:RCP_def} is a challenging optimisation problem since the cardinality of $\D$ is infinite and the optimisation has infinite number of constraints. In contrast, the SCP~\eqref{eq:SCP_def} is a convex optimisation with finite number of constraints for which efficient optimisation techniques are available \cite{boyd2004convexoptimization}. 
The following theorem provides a sample complexity result for connecting the optimal solution of the $SCP_{\gamma}$ to that of the RCP. 

\begin{theorem}[\cite{Esfahani2015scp}]
\label{thm:sample_complexity}
Assume that the mapping $d\mapsto g(\theta, d)$ in \eqref{eq:RCP_def} is Lipschitz continuous uniformly in $\theta\in T$ with Lipschitz constant $L_d$ 
and let $h\colon [0,1]\rightarrow \reals_{\ge0}$ be a strictly increasing function such that
	\begin{equation}\label{eq:inequality1}
		\prob(\ball_\varepsilon(d))\geq h(\varepsilon),
	\end{equation}
	for every $d\in\D$ and $\varepsilon\in[0,1]$. 
Let $\theta^\ast_{RCP}$ be the optimal solution of the RCP~\eqref{eq:RCP_def} and $\theta^\ast_{SCP}$  the optimal solution of $SCP_{\gamma}$~\eqref{eq:SCP_def} with
\begin{equation}\label{eq:relaxing_constant}
		\gamma =  L_dh^{-1}(\varepsilon)
	\end{equation}
computed by taking $N$ i.i.d. samples $(d_i)_{i=1}^N$ from $\prob$.
Then $\theta^\ast_{SCP}$ is a feasible solution for the RCP with confidence $(1-\beta)$ if the number of samples $N \geq N(\varepsilon,\beta)$, where 
	\begin{equation}\label{eq:sample_complexity}
		N(\varepsilon,\beta):=min\left\{N\in\nats\,\Big|\,\sum_{i=0}^{q-1}     {N\choose i}\varepsilon^i(1-\varepsilon)^{N-i}\leq \beta\right\},
	\end{equation}
	with $q$ being the dimension of the decision vector $\theta\in T$.
\end{theorem}




\section{Data-Driven Abstraction}
\label{sec:main_results}
In this section, we first discuss the steps required for model-based abstraction of control systems. We then show how this can be formulated as an RCP and present its associated SCP. 
Finally, we use the connection between the  RCPs and SCPs in Theorem~\ref{thm:sample_complexity} to provide a lower bound for number of required samples to certify a desired confidence.
The simplifying assumption used in this section is that samples from the \emph{nominal trajectories} of the system $\Sigma$ in also available in the form of $\{(x_k,u_k,x'_k)\,|\, x'_k = \Sol(x_k,u_k,0), \,k=1,2,\ldots,N\}$. We discuss in the next section how this assumption can be relaxed by modifying the inequality of the growth bound.




\subsection{Growth bound for reachable sets}
Consider a control system $\Sigma = (X, x_{\init}, U, W, f)$ with the disturbance set $W=[-\bar w,\bar w]$ for some vector $\bar w\in\reals_{\ge 0}^n$. 
 Let $\eta_x$ and $\eta_u$ be discretisation parameters for the state and input spaces $X$ and $U$ used to construct $\Xh$ and $\Uh$ of sizes $n_x$ and $n_u$, respectively. 
 The first step of ABCD is to compute a finite abstraction $\widehat\Sigma = (\Xh,\Uh,\fh)$ using overapproximations of the reachable sets for every pair of abstract state and input. The reachable set for every pair $(\xh,\uh)\in\Xh\times\Uh$ is defined as
\begin{equation*}
Reach(\xh,\uh):=\{ x'\in\Phi(x,\uh)\mid x\in\Omega_{\eta_x}(\xh)\}.
\end{equation*}
The set $Reach(\xh,\uh)$ is usually overapproximated using the growth bound of the system dynamics \cite{reissig2016feedback}. 
\begin{definition}
\label{def:growth}
The growth bound of a control system $\Sigma$ with abstract state and input spaces $\Xh,\Uh$ is a function $\kappa\colon\reals^n_{\ge0}\times \Xh\times\Uh\rightarrow\reals^n_{\ge 0}$ that satisfies
\begin{align}
	&|\Sol(x,\uh,w)-\Sol(\xh,\uh,0)|\leq \kappa(|x-\xh|,\xh,\uh) \label{eq:growth_bound_def1}\\
	&\quad\forall \xh \in \Xh,\: \forall \uh \in \Uh,\: \forall x\in\ball_{\eta_x}(\xh),\:\forall w\in W\nonumber.
\end{align}
\end{definition}
Note that $\Sol(\xh,\uh,0)$ is the nominal (disturbance-free) trajectory of the system. Using this definition, for every abstract state-input pair $(\xh,\uh)\in\Xh\times\Uh$, the reachable set $Reach(\xh,\uh)$ is overapproximated with a ball centered at $z(\xh,\uh):=\Sol(\xh,\uh,0)$ with radius $\lambda(\xh,\uh):=\kappa(\eta_x,\xh,\uh)$. 

When the system dynamics are known, it is shown in \cite{reissig2016feedback} that the growth bound can be computed as
\begin{equation}
\label{eq:reissig_growth_bound}
	\kappa(r,\xh,\uh)=e^{L(\uh)\tau}r+\int_{0}^{\tau}e^{L(\uh)s}\bar w ds,
\end{equation}
for all $r\in\reals_{\ge 0}^n$, $\xh\in\Xh$, and $\uh\in\Uh$,
where $L\colon \Uh\rightarrow\reals^{n\times n}$ is a matrix such that the entries of $L(\uh)$ satisfy the following inequality for all $x\in X$:
\begin{equation}
\label{eq:Jacob_GB}
	L_{i,j}(\uh)\geq \left\{
	\begin{array}{ll}
		D_jf_i(x,\uh) & i=j \\
		|D_jf_i(x,\uh)| & i\neq j,
	\end{array}
	\right.
\end{equation}
for all $i,j\in\{1,2,\ldots,n\}$,
where $f_i(x,u)$ is the $i^{\text{th}}$ element of the vector field $f(x,u)$ and $D_jf_i$ is its partial derivative with respect to the $j^{\text{th}}$ element of $x$.

\subsection{SCP for the computation of growth bound}
When the model of the system is unknown, the matrix $L(\uh)$ defined using \eqref{eq:Jacob_GB} is not computable, thus the growth bound in \eqref{eq:reissig_growth_bound} is not available.
To tackle this bottleneck, we use the parameterisation
\begin{equation}
\label{eq:data_driv_GB}
	\kappa(\theta)(r,\xh,\uh) := \theta_1(\xh,\uh) r+\theta_2(\xh,\uh), \forall r\in\reals_{\ge 0}^n, \xh\in\Xh,\uh\in\Uh,
\end{equation}
where $\theta_1\in\reals^{n\times n}$ and $\theta_2\in\reals^{n}$. We denote by $\theta\in\reals^{n^2+n}$ the concatenation of columns of $\theta_1$ and $\theta_2$.
\begin{remark}
The parameterised growth bound in \eqref{eq:data_driv_GB} is linear with respect to $r$ similar to \eqref{eq:reissig_growth_bound}, but is more general and less conservative by allowing $\theta_1,\theta_2$ to depend on $\xh$ (i.e., they are defined locally for each abstract state).
\end{remark}



\begin{theorem}\label{thm:RCP_GB}
The inequality \eqref{eq:growth_bound_def1} with the parameterised growth bound \eqref{eq:data_driv_GB} can be written as the robust convex program
\begin{align}
\label{eq:RCP_ABCD_F}
\begin{cases}
	\,\,\min_\theta 
	c^\top\theta\\
	\,\,s.t.\; 0\le \theta\le \bar\theta, \text{ and } \forall x\in\ball_{\eta_x}(\xh),\:\forall w\in W,\\ 
	\qquad|\Sol(x,\uh,w)-\Sol(\xh,\uh,0)|- \kappa(\theta)(|x-\xh|,\xh,\uh)\leq 0,
	\end{cases}
\end{align}
where $c=[1,1,\dots,1]\in\reals^{n^2+n}$ and $\bar\theta$ is a sufficiently large positive vector. 
\end{theorem}
\begin{proof}

We first show that the optimisation \eqref{eq:RCP_ABCD_F} is in fact a robust convex programme. Let $\D = \ball_{\eta_x}(\xh)\times W$ be the uncertainty space and
$$g(\theta,x,w):=|\Sol(x,\uh,w)-\Sol(\xh,\uh,0)|- \kappa(\theta)(|x-\xh|,\xh,\uh)$$
for all $x\in\ball_{\eta_x}(\hat x)$ and $w\in W$ and fixed $(\xh,\uh)\in\Xh\times\Uh$. We need to show that $g$ is convex in $\theta$ for each 
$(x,w)\in \D$ and bounded in $(x, w)$ for every $\theta\in [0,\bar\theta]$. The convexity holds due to the parameterisation of $\kappa(\theta)$ in \eqref{eq:data_driv_GB} being linear with respect to the optimisation variables in $\theta$. The boundedness holds due to the set $\D$ being compact and trajectories of the system being continuous.

We note that any feasible solution for the optimisation \eqref{eq:RCP_ABCD_F} gives a function $\kappa$ that satisfies the inequality \eqref{eq:growth_bound_def1} for $\Sigma$. Such a system will also have a growth bound of the form \eqref{eq:reissig_growth_bound} that is a feasible solution for \eqref{eq:RCP_ABCD_F}. To see this, we show that 
$\theta_1=e^{L(\uh)\tau}$ and $\theta_2=\int_{0}^{\tau}e^{L(\uh)s}\bar wds$ are always non-negative.
By definition, all the entries of $L(\uh)$ are non-negative except the diagonal entries. We decompose this matrix as $L(\uh)=Q+D$, where $D$ is a diagonal matrix with all diagonal entries equal to the constant $max_{i} \sum_{j} |L_{i,j}(\uh)|$ and $Q=L(\uh)-D$ is a sub-stochastic matrix as (i) its non-diagonal entries are non-negative ($Q_{i,j}=L_{i,j}(\uh)\ge o$ for $i\neq j$), (ii) its diagonal entries are non-positive ($Q_{i,i}\le 0$), and finally (iii) all of its row sums are non-positive. Note that $D$ is a multiple of identity matrix and therefore, 
	$D Q=Q D$ and $e^{(Q+D)\tau}=e^{Q\tau}e^{D\tau}$. Further, we define the matrix 
	\begin{equation*}\bar Q=
	\begin{bmatrix}
		Q & \vdots & -Q\mathbf{1}\\
		\dots & \dots & \dots\\
		\boldsymbol{0}^\top &\vdots & 0
	\end{bmatrix},
	\end{equation*}
	where $\boldsymbol{0}$ and $\mathbf 1$ represent $n-$dimensional vectors with all entries equal to zero and one, respectively. Note that $\bar Q$ is a stochastic matrix since $\bar Q_{i,i}=-\sum_{j\neq i}\bar Q_{i,j}$ for every $1\leq i\leq n+1$ and $\bar Q_{i,j}\geq0$ for $i\neq j$. Therefore, matrix $\bar Q$ correspond to the transition probability matrix of a continuous-time Markov chain with state space $\set{1,2,\ldots,n+1}$ 
	(see, e.g., \cite{katoen2008book} for more details). Therefore, the entry $(i,j)$ of $e^{\bar Q\tau}$ is the probability that the Markov chain reaches the $j^{\text{th}}$ state from the $i^{\text{th}}$ state at time $\tau$, which is a non-negative quantity. Further, we have
	$$
	e^{\bar Q\tau}=\begin{bmatrix}
		e^{Q\tau} & \vdots & \mathbf{1}-e^{Q\tau}\mathbf{1}\\
		\dots & \dots & \dots\\
		\mathbf{0}^\top &\vdots & 1
	\end{bmatrix}.
	$$
	  Therefore, $e^{Q\tau}$ is non-negative, which gives $e^{L(\uh)\tau} = e^{Q\tau} e^{D\tau}$ since $e^{D\tau}\ge 0$. This naturally results in $\theta_1$ and $\theta_2$ being non-negative as the integral of non-negative functions.
\end{proof}

\smallskip

To construct the $\text{SCP}_\gamma$ associated with the RCP~\eqref{eq:RCP_ABCD_F}, we fix $\xh\in\Xh$ and $\uh\in\Uh$, consider a uniform distribution on the space $\D = \ball_{\eta_x}(\xh)\times W$
and obtain $N$ i.i.d. sample trajectories $\mathcal S_N = \set{(x_i,\uh,x'_i)\,|\, x'_i \in \Phi(x_i,\uh), i=1,2,\ldots,N}$. Note that every $x'_i$ corresponds to a random disturbance $w_i\in W$. 
The  $\text{SCP}_\gamma$ is
\begin{align}
\label{eq:SCP_ABCD_F}
\begin{cases}
	\,\,\min_{\theta} c^\top\theta\\
	\,\,s.t.\;0\le \theta\le\bar\theta \text{ and } \forall i\in\{1,\dots, N\},\\ 
	\qquad |x'_i-x'_{nom}|- \theta_1(\xh,\uh) |x_i-\hat x|+\theta_2(\xh,\uh)+\gamma\leq 0,
	\end{cases}
\end{align}
where $x'_{nom}:=\Sol(\xh,\uh,0)$ and $\gamma\in\reals_{\ge 0}$.

\begin{theorem}
\label{thm: sound_data_driven_abstraction}
For any $\xh\in\Xh$ constructed with discretisation size $\eta_x$, any $\uh\in\Uh$, and the disturbance set $W = [-\bar w,\bar w]$, the optimal solution of \eqref{eq:SCP_ABCD_F} gives a growth bound for the system $\Sigma$ corresponding to $(\xh,\uh)$ with confidence $(1-\beta)$, when the number of samples $N\ge N(\varepsilon,\beta)$ and
 \begin{equation}
 \label{eq:bias_def}
\gamma=4L_\Sol(\uh)\sqrt[2n]{\varepsilon\prod_{i=1}^{n}\eta_x(i)\prod_{i=1}^{n}\bar w(i)},
\end{equation}
where $\varepsilon\in[0,1]$, 
$n$ is the dimension of the state space and $L_\Sol(\uh)$ is the Lipschitz constant of the system trajectories $\Sol(x,\uh,w)$ with respect to $(x,w)$.
\end{theorem}

\begin{proof}
We apply Theorem~\ref{thm:RCP_GB} to the RCP \eqref{eq:RCP_ABCD_F} for fixed $\xh\in\Xh$ and $\uh\in\Uh$. 
Define
\begin{align}
\label{eq:gf}
g(\theta,x,w):=\max\{|\Sol&(x,\uh,w)-\Sol(\xh,\uh,0)|\\
&- \theta_1(\xh,\uh)|x-\xh| - \theta_2(\xh,\uh)\},\nonumber
\end{align}
where the $\max\{\cdot\}$ is applied to the elements of its argument that belongs to $\mathbb R^n$.
Since the distribution on $\D = \ball_{\eta_x}(\xh)\times W$ is uniform, we choose
$$h(\varepsilon) = \prob(\ball_\varepsilon(d)) = \frac{(\varepsilon/2)^{2n}}{\prod_{i=1}^{n}\eta_x(i)\prod_{i=1}^{n}\bar w(i)}$$
to satisfy the inequality \eqref{eq:inequality1}. 
Note that $h(\varepsilon)$ gives the probability of choosing a point within the $2n-$ball $\ball_\varepsilon(d)$ uniformly at random. We use Equation~\eqref{eq:relaxing_constant} as $\gamma = L_d h^{-1}(\varepsilon)$ to get the value of $\gamma$ in \eqref{eq:bias_def}. It only remains to show that $g(\theta,x,w)$ is Lipschitz continuous with constant $L_d=2L_\Sol(\uh)$. Note that $L_\Sol(\uh)$ is the Lipschitz constant of $\Sol(x,\uh,w)$ with respect to $(x,w)$, and satisfies
\begin{equation}
\label{eq:Lipschitz}
\|\Sol(x,\uh,w)-\Sol(x',\uh,w')\|\le L_\Sol(\uh) \|(x,w) - (x',w')\|
\end{equation}
for all $x,x'\in \ball_{\eta_x}(\xh)$ and $w,w'\in W$.
Since $\|\theta_1(\xh,\uh)\|$ can be bounded by $L_\Sol(\uh)$, we get that
\begin{align*}
	& \|g(\theta,x,w) - g(\theta,x',w')\|\\
	& \le\|\Sol(x,\uh,w) - \Sol(x',\uh,w')\|+ \|\theta_1(\xh,\uh)\|\|x-x'\|\\
	& \le L_\Sol(\uh) \|(x,w) - (x',w')\| +  L_\Sol(\uh) \|x- x'\|\\
	& \le 2L_\Sol(\uh) \|(x,w) - (x',w')\|,
\end{align*}
Therefore, $g(\theta,x,w)$ is Lipschitz continuous with constant $2L_\Sol(\uh)$.
This completes the proof.
\end{proof}

\begin{remark}
The value of $\gamma$ in \eqref{eq:bias_def} depends on the Lipschitz constant $L_\Sol$. We provide an algorithm in the next subsection for estimating this constant using sampled trajectories of the system. Note that as the above proof shows, the estimated quantity $\theta_1 = L_\Sol\mathbf 1_{n\times n}$ can be used to construct the abstraction, but this would give conservative results without any formal guarantee. We will demonstrate this observation on a case study in Section~\ref{sec:experiments}. 
\end{remark}

\begin{corollary}
\label{cor:conf}
The abstract model constructed using the growth bounds as solutions of $\text{SCP}_\gamma$ with confidence $(1-\beta)$ for state-input pairs $(\xh,\uh)\in \Xh\times\Uh$ is a valid abstract model for $\Sigma$ with confidence at least $(1-n_xn_u\beta)$, where $n_x$ and $n_u$ are respectively the cardinality of $\Xh$ and $\Uh$.
\end{corollary}
\begin{proof}
Denote the optimal solution of $\text{SCP}_\gamma$ in \eqref{eq:SCP_ABCD_F} by $\theta^\ast$. The ball centered at $z(\xh,\uh):=x'_{nom}$ with radius $\lambda(\xh,\uh)=\kappa(\theta^\ast)(\eta_x,\xh,\uh)+\gamma$ is a valid overapproximation of the reachable set from the state-input pair $(\xh,\uh)$ with confidence at least $1-\beta$. Since the number of pairs $(\xh,\uh)$ is $n_xn_u$, the chance of getting an invalid growth bound in at least one instance of $\text{SCP}_\gamma$ is bounded by $n_xn_u\beta$. Therefore, we get a sound abstraction that truly overapproximates the behaviour of the system with confidence $(1-n_xn_u\beta)$.
\end{proof}
\begin{remark}
	 The parameter $\varepsilon\in[0,1]$ gives a tradeoff between the required number of samples and the level of conservativeness applied to the SCP. Smaller $\varepsilon$ results in a larger number of sample trajectories, but reduces the value of $\gamma$ in \eqref{eq:bias_def} (less conservative constraints in the SCP and higher chance of finding a feasible solution). In contrast, larger $\varepsilon$ results in a smaller number of sample trajectories but increases the value of $\gamma$.
\end{remark}
\begin{remark}
The quantity $2n$ used in \eqref{eq:bias_def} is in fact the dimension of the sample space $\D = \Omega_{\eta_x}(\hat x)\times W$. If the system does not have any disturbance (i.e., the system can be modeled as an ODE having deterministic trajectories), the sample space will be $\D = \Omega_{\eta_x}(\hat x)$ and its dimension $n$ can be used in \eqref{eq:bias_def}: $\gamma=4L_\Sol(\uh)\sqrt[n]{\varepsilon\prod_{i=1}^{n}\eta_x(i)}$.
 This will substantially reduce the number of required sample trajectories. Similarly, if the disturbance does not affect some of the state equations, $2n$ can be replaced by $(n+q)$ where $q$ is the dimension of the disturbance set considered as a non-zero measure set.
\end{remark}


Algorithm~\ref{alg:main_abstraction} uses the result of Corollary~\ref{cor:conf} to provide an algorithmic solution for Problem~\ref{prob:abstraction}. 
This algorithm receives a confidence parameters $\beta$, divides it by the cardinality of $\Xh\times\Uh$ (i.e., $n_xn_u$), 
computes the growth bounds for each pair $(\xh,\uh)\in \Xh\times\Uh$ using the $\text{SCP}_\gamma$ in \eqref{eq:SCP_ABCD_F} with confidence $1-\beta/(n_xn_u)$, and constructs the abstraction using these growth bounds. 

\begin{algorithm}
	\caption{Data-Driven Abstraction}
	\label{alg:main_abstraction}
	\KwData{$(X, U, W)$ of a control system $\Sigma$, confidence $\beta$, discretisation parameters $\eta_x$, $\eta_u$}
	\label{alg:data-drivenABCD}
	Compute the finite state and input sets $\Xh$ and $\Uh$ using $\eta_x$, $\eta_u$\;
	Define  $n_x$ and $n_u$ as cardinalities of $\Xh$ and $\Uh$\;
	Choose $\varepsilon\in[0,1]$\; 
	Set $ N=N(\varepsilon,\frac{\beta}{n_xn_u})$ using Eq. \eqref{eq:sample_complexity}\;
	Compute $\gamma$ using Eq.~\eqref{eq:bias_def}\;
	
	\For{$\xh\in\Xh$}{
		\For{$\uh\in\Uh$}{
			$\hat f(\hat x,\hat u)=\emptyset$\;
			Consider the uncertainty space $\D = \ball_{\eta_x}(\xh)\times W$\;
			Select $N$ i.i.d sample trajectories using uniform distribution over $\D$\;
			Simulate the nominal trajectory $(\xh,\uh,x'_{nom})$\;
			Solve the $\text{SCP}_\gamma$ \eqref{eq:SCP_ABCD_F} to get the optimiser $\theta^\ast(\xh,\uh)$\;
			$z\leftarrow x'_{nom}$\;
			 $\lambda\leftarrow\kappa(\theta^\ast)(\eta_x,\xh,\uh)+\gamma$\;
			Find all states $\xh'\in\Xh$ for which $\ball_{\eta_x}( \xh')\cap\ball_{\lambda}(z)\neq\emptyset$ and add them to $\hat f (\hat x,\hat u)$\;
		}
	}
	\KwResult{$\Sigmah = (\hat X, \hat U, \hat f)$ as a finite abstraction of $\Sigma$ with confidence $(1-\beta)$, $\theta^\ast(\xh,\uh)$ as a growth bound for $\xh\in\Xh, \uh\in\Uh$}
\end{algorithm}







The finite abstraction $\Sigmah$ constructed by Algorithm~\ref{alg:main_abstraction} is a valid abstraction for $\Sigma$ with confidence $(1-\beta)$. This means any controller $\Ch$ synthesised on $\Sigmah$ and refined to a controller $C$ for $\Sigma$ will satisfy the desired specification with confidence $(1-\beta)$ on the closed loop system $\Sigma\parallel C$. In the next section, we extend our approach to make it suitable for abstraction refinement in case there is no controller $\Ch$ satisfying the specification due to the conservatism of the approach. 

\subsection{Lipschitz Constant Estimation}
For estimating the Lipschitz constant $L_\Sol$ in \eqref{eq:Lipschitz}, we estimate an upper bound for the fraction
\begin{equation*}
\Delta(\uh) := \frac{\|\Sol(x,\uh,w)-\Sol(x',\uh,w')\|}{\|(x,w) - (x',w')\|}
\end{equation*}
that holds for all $x,x'\in X$ and $w,w'\in W$. We follow the line of reasoning in \cite{wood1996estimation,weng2018evaluating} and use the extreme value theory for the estimation.


Let us fix a $\delta > 0$ and assign uniform distribution to the pairs $(x,w)$ and $(x',w')$ over the domain
\begin{equation}
\label{eq:domain}
\{x,x'\in X,\, w,w'\in W \text{ with } \|(x,w) - (x',w')\| \leq \delta\}.
\end{equation}
Then $\Delta(\uh)$ is a random variable with an unknown cumulative distribution function (CDF). Based on the assumption of Lipschitz continuity of the system, the support of the distribution of $\Delta(\uh)$ is bounded from above, and we want to estimate an upper bound for its support. We take $\mathfrak n$ sample pairs $(x,w)$ and $(x',w')$, and compute $\mathfrak n$ samples $\Delta_1,\Delta_2,\ldots,\Delta_{\mathfrak n}$ for $\Delta(\uh)$. The CDF of $\max\{\Delta_1,\Delta_2,\ldots,\Delta_{\mathfrak n}\}$ is called the limit distribution of $\Delta(\uh)$.
Fisher-Tippett-Gnedenko theorem says that if the limit distribution exists, it can only be one of the three family of extreme value distributions -- the Gumbel class, the Fr\'echet class, and the reverse Weibull class.
These CDF's have the following forms:
\begin{align*}
& \text{Gumbel class: } \quad G(s) = \exp\left[-\exp\left[\frac{s-a}{b}\right]\right],\, s \in \mathbb{R}\\
& \text{Fr\'echet class: }\quad G(s) = \begin{cases} 
0 & \text{ if } s < a \\
\exp\left[-[\frac{s - a}{b}]^{-c}\right] & \text{ if } s \leq a
\end{cases}\\
& \text{Reverse Weibull class: }
G(s) = \begin{cases}
\exp\left[-[\frac{a - s}{b}]^{c}\right] & \text{ if } s < a \\
1 & \text{ if } s \leq a \\
\end{cases}
\end{align*}
where $a\in \mathbb{R}, b > 0, c > 0$ are respectively the location, scale and shape parameters of the distributions.

Among the above three distributions, only the reverse Weibull class has a support bounded from above. Therefore, the limit distribution of $\Delta(\uh)$ will be from this class and the location parameter $a$ is such an upper bound. As a result, we can estimate the location parameter of the limit distribution of $\Delta(\uh)$ to get an estimation of the Lipschitz constant.

The approach is summarised in Algorithm~\ref{alg:ld_est}. The most inner loop computes samples of $\Delta(\uh)$. The middle loop computes samples of $\max\{\Delta_1, \ldots, \Delta_{\mathfrak n}\}$. The outer loop estimates the Lipschitz constant for each $\uh$ by fitting a reverse Weibull distribution.
 
\begin{algorithm}
	\caption{Lipschitz Constant Estimation}
	\label{alg:ld_est}
	\KwData{$(X, U, W)$ of a control system $\Sigma$, abstract input space $\Uh$}
	Select number of samples $\mathfrak n$ and $\mathfrak m$ for the estimation\\
	Select $\delta>0$\\
	\For{$\uh \in \Uh$}{
		\For{$j = 1 : \mathfrak m$}{
			\For{$i = 1 : \mathfrak n$}{
				Sample pairs $(x,w), (x',w')$ uniformly from the domain in \eqref{eq:domain}\\
				Run $\Sigma$ to get trajectories $\Sol(x,\uh,w)$ and $\Sol(x',\uh,w')$\\
				Compute $\Delta_i := \frac{\|\Sol(x,\uh,w)-\Sol(x',\uh,w')\|}{\|(x,w) - (x',w')\|}$
			}
			$\Gamma_j := \max\{\Delta_1, \ldots, \Delta_{\mathfrak n}\}$
		}
		Fit a reverse Weibull distribution to the sample set $\{\Gamma_1,\Gamma_2,\ldots,\Gamma_{\mathfrak m}\}$\\
		$L_\Sol(\uh)$ is the location parameter of the fitted distribution
		
	}
	\KwResult{Estimated value of $L_\Sol(\uh)$ for all $\uh\in\Uh$}
\end{algorithm}



\section{Synthesis via Abstraction Refinement}
\label{sec:refinement}


The data-driven synthesis discussed in Section~\ref{sec:main_results} inherits the soundness property from the ABCD approach: they both work with overapproximations of the dynamics and may not return a controller despite one may exists. Therefore, there is a need for refining the abstraction in order to check for controllers using less conservative abstractions. While the method of Section~\ref{sec:main_results} is good for a given fixed discretisation parameter $\eta_x$, it is not suitable for reducing $\eta_x$, which requires re-computing all local parameters of the growth bounds $\theta_1(\xh,\uh),\theta_2(\xh,\uh)$. Another shortcoming of the method is related to the data collection: the nominal trajectories of the system should be available and are used in the constraints of the SCP.
In this section, we discuss an extension of the approach of Section~\ref{sec:main_results}, in order to
\begin{itemize}
\item enable reducing $\eta_x$ without the need for re-computing the growth bound, and
\item relax the assumption of having access to the nominal trajectories of the system.
\end{itemize}



Let us define a modified growth bound as a function $\kappa_e\colon\reals^n_{\ge0}\times \Xh\times\Uh\rightarrow\reals^n_{\ge 0}$ that is strictly increasing in its first argument and satisfies
	\begin{align}
		&|\Sol(x_1,\uh,w_1)-\Sol(x_2,\uh,w_2)|\leq \kappa_e(|x_1-x_2|,\xh,\uh)\nonumber\\
		&\,\,\forall \xh \in \Xh,\: \forall \uh \in \Uh,\: \forall x_1,x_2\in\ball_{\eta_x}(\xh),\:\forall w_1,w_2\in W.\label{eq:growth_bound_def2}
	\end{align}
This definition is more conservative than \eqref{eq:growth_bound_def1} in comparing trajectories under two arbitrary disturbances, and we always have that $\kappa_e$ satisfies \eqref{eq:growth_bound_def1}.
Using this new definition, for every pair of abstract state and input $(\xh,\uh)$, the corresponding overapproximation of the reach set can be computed as a ball centred at \emph{any} $z(\xh,\uh)\in\Phi(\xh,\uh)$ with radius $\lambda(\xh,\uh)=\kappa_e( \eta_x,\xh,\uh)$.



we choose a parametrisation for $\kappa_e$ similar to \eqref{eq:data_driv_GB}, i.e.,
\begin{equation}
	\label{eq:REGB_parametrization}
	\kappa_e(\theta)(r,\xh,\uh) = \theta_1(\xh,\uh) r+\theta_2(\xh,\uh),
\end{equation}
where  $r\in\reals_{\ge 0}$, $\theta_1\in\reals^{n\times n}$, $\theta_2\in\reals^{n}$, and $\theta\in\reals^{n^2+n}$ is constructed by concatenating columns of $\theta_1$ and $\theta_2$.
The SCP associated with this growth bound is constructed by considering a uniform distribution over $\ball_{\eta_x}(\xh)\times W$ and obtain $2N$ i.i.d. sample trajectories $\mathcal S_{2N}=\set{(x_i,\uh_i,x'_i)\,|\, x'_i \in \Phi(x_i,\uh), i=1,2,\ldots,2N}$ so that every $x_i'$ corresponds to a random disturbance $w_i\in W$. The modified $\text{SCP}_\gamma$ is defined as
\begin{equation}
	\label{eq:SCP_ABCD_AA}
	\begin{cases}
		\min c^\top\theta\\
		s.t.\;0\leq\theta\leq\bar\theta\;\text{and }\forall i\in\{1,\dots, N\} \nonumber\\
		|x'_{2i-1}-x'_{2i}|- \theta_1(\xh,\uh)|x_{2i-1}-x_{2i}|-\theta_2(\xh,\uh)+\gamma	\leq 0 
	\end{cases}
\end{equation}
where $c=[1,1,\dots,1]\in\reals^{n^2+n}$ is a constant vector, $\bar\theta\in\reals_{>0}^{n^2+n}$ is sufficiently large, and $\gamma\ge 0$. 

\begin{theorem}
	\label{thm:refined_abstraction}
	For any $\xh\in\Xh$ constructed with the discretisation size $\eta_x$, any $\uh\in\Uh$, and the disturbance set $W = [-\bar w,\bar w]$, the optimal solution of \eqref{eq:SCP_ABCD_AA} gives a growth bound for the system $\Sigma$ corresponding to $(\xh,\uh)$ that satisfies \eqref{eq:growth_bound_def2} with confidence $(1-\beta)$, when the number of samples $2N\ge N(\varepsilon,\beta)$ and
	\begin{equation}\label{eq:bias_def_A}
		\gamma=8L_\varphi\sqrt[4n]{ \varepsilon\left[\prod_{i=1}^{n}\eta_x(i)\prod_{i=1}^{n}\bar w(i)\right]^2},
	\end{equation}
	where $\varepsilon\in[0,1]$,  
	,$n$ is the dimension of the state space, 
	and $L_\Sol(\uh)$ is the Lipschitz constant of the system trajectories $\Sol(x,\uh,w)$ with respect to $(x,w)$.
\end{theorem}
\begin{proof}
	The proof of this theorem is similar to that of Theorem~\ref{thm: sound_data_driven_abstraction}.
	Define \begin{align*}
		g(\theta,x_1,w_1,x_2,w_2)\!:=\!\max\{&|\Sol(x_1,\uh,w_1)-\Sol(x_2,\uh,w_2)|\\
		&- \theta_1(\xh,\uh)|x_1-x_2| - \theta_2(\xh,\uh)\}.
		\end{align*}
		To satisfy the inequality \eqref{eq:inequality1}, we can choose
		\begin{equation*}
		h(\varepsilon) = \prob(\ball_\varepsilon(d)) = \frac{(\varepsilon/2)^{4n}}{[\prod_{i=1}^{n}\eta_x(i)\prod_{i=1}^{n}\bar w(i)]^2},
		\end{equation*}
	since the distribution on $(\ball_{\eta_x}(\xh)\times W)^2$ is uniform. 
	Using Equation~\eqref{eq:relaxing_constant}, we have $\gamma = L_d h^{-1}(\varepsilon)$. In order to prove that $\gamma$ takes the value in \eqref{eq:bias_def_A}, we must show that $g$ is Lipschitz continuous with constant $L_d=4L_\Sol(\uh)$. Bounding $\|\theta_1(\xh,\uh)\|$ by $L_\Sol$, for all $(x_{1},w_{1},x_{2},w_{2})$ and $(x_{1}',w_{1}',x_{2}',w_{2}')$ we have
		\begin{align*}
			& \|g(\theta,x_{1},w_{1},x_{2},w_{2}) - g(\theta,x_{1}',w_{1}',x_{2}',w_{2}')\|\\
			& \le\|\Sol(x_{1},\uh,w_{1}) - \Sol(x_{1}',\uh,w_{1}')\|\\
			&+\|\Sol(x_{2},\uh,w_{2}) - \Sol(x_{2}',\uh,w_{2}')\|\\
			 &+\|\theta_1(\xh,\uh)\|(\|x_{1}-x_{1}'\|+\|x_{2}-x_{2}'\|)\\
			& \le 4L_\Sol(\uh) \|(x_{1},w_{1},x_{2},w_{2}) - (x_{1}',w_{1}',x_{2}',w_{2}')\|.
	\end{align*}
Therefore, $g$ is Lipschitz continuous with constant $4L_\Sol(\uh)$.
	This completes the proof.
\end{proof}
A statement similar to Corollary~\ref{cor:conf} holds for the growth bound computed using \eqref{eq:SCP_ABCD_AA}.

\section{Experimental Evaluation}
\label{sec:experiments}
To demonstrate our approach, we apply it to a DC-DC boost converter and a path planning problem. 
These case studies are taken from \cite{rungger2016scots, girard2009approximately} and will be used as black-box models to generate sample trajectories.
We also introduce a case study from power systems based on \cite{MaFan2016}, that is implemented in the Power System Toolbox (PST) \cite{ChowCheung1992}. We will will use trajectories from the black-box reduced model of the 30 state power system model. We apply our approach to construct finite abstractions of these systems and employ SCOTS \cite{rungger2016scots} to design controllers.
%
Our algorithms are implemented in C++ on a 64-bit Linux cluster machine with two Intel Xeon E5 v2 CPUs, 1866 MHz, and 50GB RAM. 

\subsection{DC-DC boost converter}
The objective in the DC-DC boost converter problem is to design a controller to enforce a reach and stay specification. The DC-DC boost converter can be modelled as a two dimensional linear switching system with two functional modes. The state vector of the system at time $t\in\reals_{\ge 0}$ is $x(t) = (i_l(t),v_c(t))$, where $i_l$ is the inductor current and $v_c$ is the capacitor voltage. The system's evolution can be controlled by selecting the appropriate mode $u(t)\in\set{1,2}$ at every time $t\in\reals_{\ge 0}$. The system's dynamics under the two modes can be represented as $\dot{x} = A_{u(t)} x(t)+ b+cw(t)$, $u\in\{1, 2\}$, with matrices $A_1,A_2,b,c$ as reported in \cite{girard2009approximately}. The state and input spaces are $X= [0.65, 1.65]\times[4.95, 5.95]$ and $U= [1, 2]$. 
The initial state is $(i_{l_0}(t),v_{c_0}(t))=(0.7, 5.4)$ and the target set is $[1.1, 1.6]\times[5.4, 5.9]$. The target set is shown in red colour in \Cref{fig:dcdc-nd}.

Our implementation results are reported in Table~\ref{tab:dcdc} for the system without disturbance ($\bar w = (0,0)$) and with disturbance bound $\bar w = (0.01,0)$. These results are obtained with discretisation parameters $\eta_x=(0.005, 0.005)$ and $\eta_u= 1$, confidence parameter $\beta = 0.01$, $\varepsilon = 0.01$ and estimation for $L_{\Sol} = 0.9935$. The resulted finite abstraction has cadinalities $n_x = 40,000$ and $n_u = 2$. The required number of sample trajectories, $N$, for each $(\xh,\uh)\in\Xh\times\Uh$ is computed using equation \eqref{eq:sample_complexity}. Runtimes and the resulting winning region sizes, $|\mathcal V|$, for the DC-DC boost converter are given in Table~\ref{tab:dcdc}.

\begin{table}
		
	\caption{Results for the DC-DC boost converter. 
	}
	\renewcommand{\arraystretch}{1.2}
	\setlength{\tabcolsep}{0.7em} 
	\resizebox{1\columnwidth}{!}{
	\begin{tabular}{l|cc|c|ccc}
		\toprule
		Case-study&\multicolumn{2}{c|}{Dimension}& Disturbance& 
		\multicolumn{3}{c}{Fixed Discretisation}\\
		
		&$X$&$U$&$W$&$N$& time (min)&$|\mathcal V|$\\
		\midrule
		\multirow{2}{*}{\rotatebox{0}{DC-DC boost converter}} &\multirow{2}{*}{\rotatebox{0}{$2$}}&\multirow{2}{*}{\rotatebox{0}{$1$}}&$\{0\}$& $1,807$&$22.2$  & $37,783$\\
		&&&$[-0.01,0.01]$& $2,285$&$30.6$  & $37,414$\\
		
		
		\bottomrule
	\end{tabular}}
	\label{tab:dcdc}
\end{table}

We have used Algorithm~\ref{alg:main_abstraction} to compute the finite-state abstraction by collecting sample trajectories of the system. Subsequently, SCOTS is used for designing the controller.
The performance of the controller is shown in \Cref{fig:dcdc-nd,fig:dcdc-d} for the system without and with the disturbance. These figures show one sample closed-loop trajectory of the system under the controllers designed by our data-driven ABCD approach. In both cases, without and with disturbance, it can be noticed from \Cref{fig:dcdc-nd,fig:dcdc-d} that our approach has been successful in finding controllers satisfying the given reach and stay specification, despite the the dynamics being unknown.

	\begin{figure}
	\includegraphics[width=0.9\columnwidth]{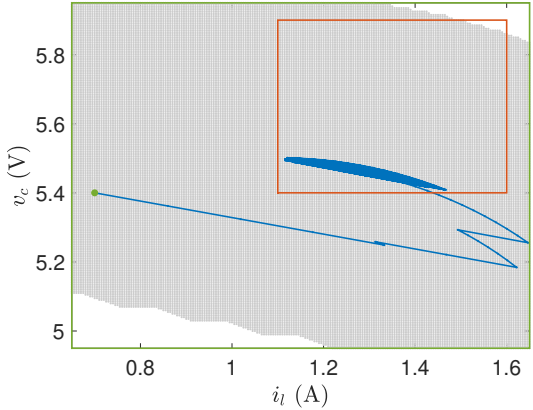}
	\caption{The closed-loop trajectory of the DC-DC boost converter with $\bar w = (0, 0)$ under the controller designed by our data-driven abstraction approach. The rectangle in red colour represents the target region and the area in grey shows the winning region of the controller.}
	\label{fig:dcdc-nd}	
	\end{figure}
	
	\begin{figure}
		\includegraphics[width=0.9\columnwidth]{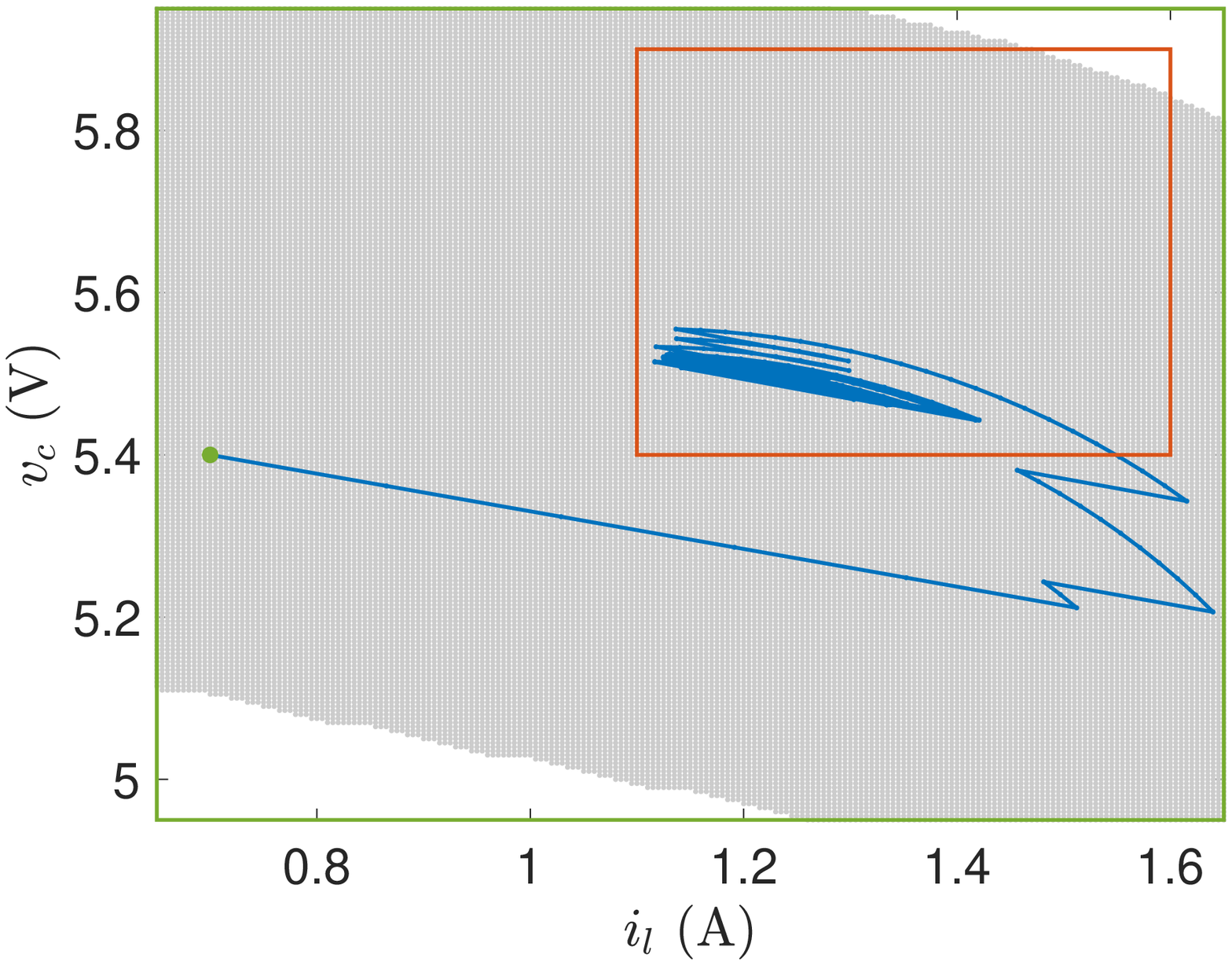}
		\caption{The closed-loop trajectory of the DC-DC boost converter with $\bar w = (0.01,0)$ under the controller designed by our data-driven abstraction approach. The rectangle in red colour represents the target region and the area in grey shows the winning region of the controller.}
		\label{fig:dcdc-d}	
	\end{figure}
	

\subsection{Path Planning Problem with Partition Refinement}
We consider a path planning problem for a vehicle that is modelled as
%
\begin{equation}
\label{eq:car}
	\begin{array}{l}
		\dot{x} = v\cos(\alpha+\theta)/\cos(\alpha)+w\\
		\dot{y} = v\sin(\alpha+\theta)/\cos(\alpha)\\
		\dot{\theta} = v\tan(\omega),
	\end{array}
\end{equation}
where the state variables $x, y, \theta$ represent the position of the vehicle in the $2$-dimensional space and the orientation of the vehicle, respectively. Inputs are $(v, \omega)$, the disturbance is $w$, and $\alpha:=\arctan(\tan(\omega)/2)$. 
The state and input spaces are $X = [0, 10]\times[0, 10]\times[-\pi-0.4, \pi+0.4]$ and $U = [-1, 1]^2$, respectively.
The goal is to find a controller to steer the vehicle from the initial state $(x_0,y_0,\theta_0) = (0, 1.2, 0)$ to the target set $ (x,y)\in [9, 9.51] \times [0, 0.51]$ while avoiding the obstacles. These obstacles are shown in blue colour in \Cref{fig:vehicle-nd,fig:vehicle-d}.

We computed the growth bounds with a coarse discretisation $\eta_x=(1.6,1.6,1.6)$ and reduced it iteratively with the factor of two. The algorithm successfully finds a controller for the system after five iterations. The implementation results are reported in Table~\ref{tab:vehicle}. These results are obtained with
$\eta_u=(0.3, 0.3)$, the confidence parameter $\beta = 0.01$, $\varepsilon = 0.01$ and estimated constant $L_{\Sol} = 1.46$. The resulted abstraction has cardinalities $n_x = 88,500$ and $n_u = 24$.
For the case of disturbance-free model we set $\bar w = (0, 0, 0)$, and for the case of dynamics with disturbance, we set $\bar w = (0.01, 0, 0)$. The required number of sample trajectories for each $(\xh,\uh)$ is computed using Equation~\eqref{eq:sample_complexity} and marked with $N$ in the table. Finally, runtimes and size of the winning regions $|\mathcal V|$ are reported.


\begin{table}
		
	\caption{Results for the path planning case study.}
	\renewcommand{\arraystretch}{1.2}
	\setlength{\tabcolsep}{0.7em} 
	\resizebox{1\columnwidth}{!}{
	\begin{tabular}{l|cc|c|ccc}
		\toprule
		Case-study&\multicolumn{2}{c|}{Dimension}& Disturbance& 
		\multicolumn{3}{c}{Abstraction Refinement}\\
		
		&$X$&$U$&$\bar w$&$N$& time (min)&$|\mathcal V|$\\
		\midrule
		\multirow{2}{*}{\rotatebox{0}{Path planning}} &\multirow{2}{*}{\rotatebox{0}{$3$}}&\multirow{2}{*}{\rotatebox{0}{$2$}}&$(0, 0, 0)$& $3,127$&$225$  & $405,493$\\
		&&&$(0.01, 0, 0)$& $4,277$&$513$  & $447,212$\\
		
		
		\bottomrule
	\end{tabular}}
	\label{tab:vehicle}
\end{table}

We have used the synthesis method based on abstraction refinement presented in Section~\ref{sec:refinement}, to construct the finite-state abstraction by collecting sample trajectories of the system. We used SCOTS to design the controller fulfilling the given specification. The performance of the controller is shown in \Cref{fig:vehicle-nd,fig:vehicle-d} for the system without and with the disturbance, respectively. These figures compare the closed-loop trajectories of the system under the controllers designed by our data-driven abstraction refinement algorithm approach (black) and by the model-based approach of SCOTS (red). Our data-driven approach successfully finds a controller for the system that satisfies the specification without the need for knowing the dynamics of the system.

	\begin{figure}
		\includegraphics[width=0.9\columnwidth]{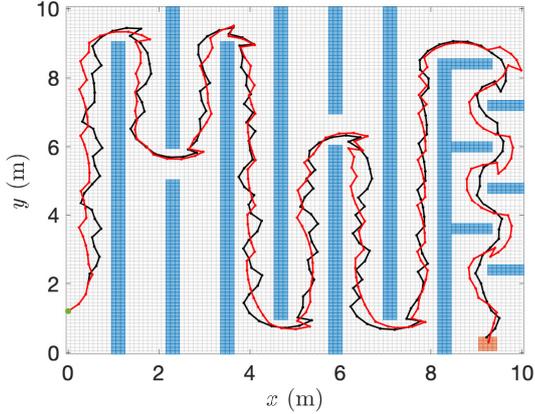}
		\caption{Comparison between the closed-loop trajectories of the system \eqref{eq:car}without disturbance under the controllers designed by our data-driven abstraction refinement approach (black) and by the model-based approach of SCOTS (red). Blue blocks represent the obstacles, the green dot represents the initial state, and the orange rectangle shows the target region.}
		\label{fig:vehicle-nd}
	\end{figure}
	\begin{figure}
	\includegraphics[width=0.9\columnwidth]{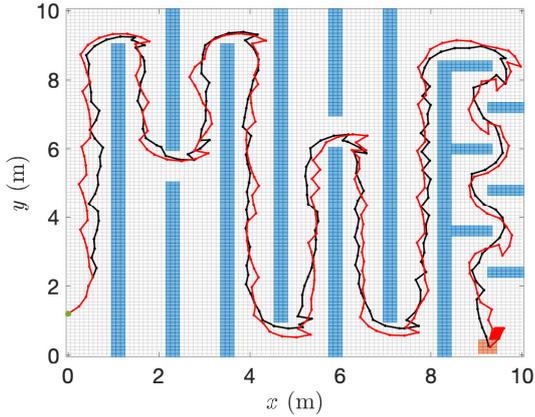}
	\caption{Comparison between the closed-loop trajectories of the system \eqref{eq:car} with disturbance bound $\bar w = (0.01, 0, 0)$ under the controllers designed by our data-driven abstraction refinement approach (black) and by the model-based approach of SCOTS (red).}
	\label{fig:vehicle-d}
	\end{figure}




\subsection{Three Area Three Machine Power System}

\begin{figure}
    \centering
    \includegraphics[width=0.9\columnwidth]{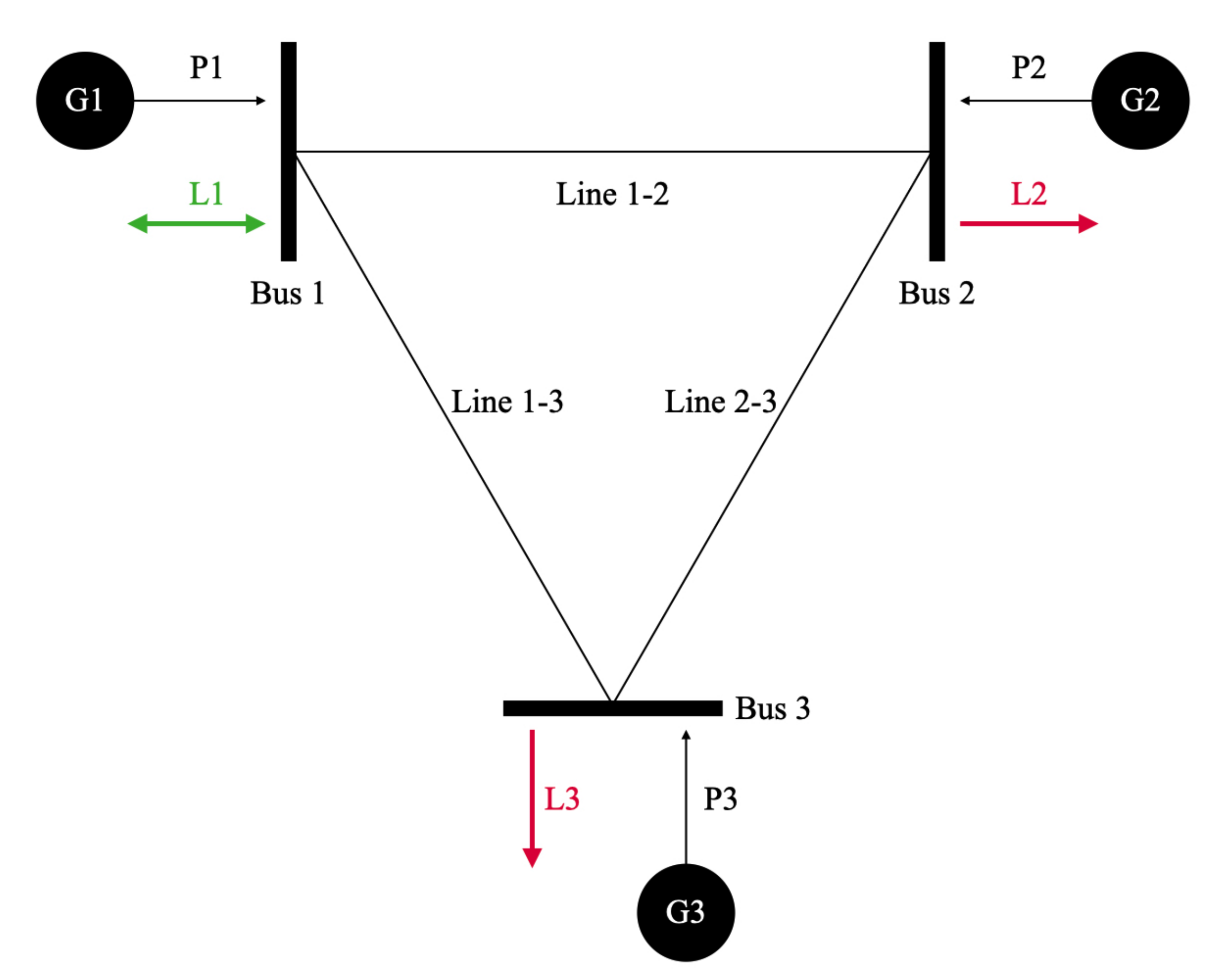}
    \caption{3A3M power system with generators (G) and loads (L). L1 represents a bidirectional load such as Electric Vehicles or Energy Storage Systems.}
    \label{fig:threearea}
  \end{figure}
  
We consider a three area three machine (3A3M) power system adapted from \cite{MaFan2016} and is shown in \Cref{fig:threearea}. The system consists of three buses, which are each connected to a power source (generator) and a load. At bus $1$ we consider a load which is bidirectional, meaning it can both draw power and inject power into the system. The loads at buses $2$ and $3$ can only draw power from the system; when these loads increase, more power will be drawn from the system, causing an imbalance between generation and consumption which may result in reduction of the network frequency. The nominal frequency of the network is set to $60$ Hz.

We consider a worst case scenario when a sudden increase occurs in the loads at buses $2$ and $3$ by $0.2$ and $0.3$ per unit (pu), respectively. The control task is for the load at bus $1$ to balance the load increase at buses $2$ and $3$ by either reducing its load or injecting power into the network.
The simulation is run using PST on a $30$ state model of this power system. Balanced realisation of the system reduces its dynamics to three states. To compute the data-driven finite abstraction, sample trajectories are gathered using a black-box approach of the reduced system representation for the original model.
The dynamics of the reduced system are given by
\begin{equation}
	\begin{array}{l}
    \dot{x} = Ax + Bu + Ew \\
        y = Cx,
        \end{array}
\end{equation}
where
\begin{equation}
    A = \begin{bmatrix}
        0.00027563  & 0 & 0\\
        0 & -0.3951 &  0.687\\
        0  &  -0.6869 & -0.016
        \end{bmatrix} \nonumber
\end{equation}

\begin{equation}
    B = \begin{bmatrix}
        0.00031166\\
        0.1359 \\
        0.0230 
        \end{bmatrix} \nonumber
\end{equation}
   
\begin{equation}
    E = \begin{bmatrix}
        0.00033103  & 0.00031244\\
        0.1309 & 0.1308\\
        0.0250  &  0.0233
        \end{bmatrix} \nonumber
\end{equation}

\begin{equation}
    C = \begin{bmatrix}
        -0.0115 & -0.2296 & 0.0412 \\
                \end{bmatrix}.
\end{equation}

The state and input spaces are $X = [-0.02,0.02]\times[-0.05,0.05]\times[-0.12, 0.12]$ and $U = [0,0.5]$. Further, we set $W =[-0.2, 0.2]\times[-0.3,0.3]$, $\eta_u = 0.025$, $\tau = 0.4$, $\eta_x = (0.0015, 0.0015, 0.0015)$, $\beta = 0.01$ and $\varepsilon = 0.01$. The resulted abstraction has $n_x = 228,480$ and $n_u = 20$. The estimated Lipschitz constant is $L_{\Sol} = 1.5715$.
The target set is given by $-0.008 < y < 0.008$ and the avoid set is given by $y < -0.015$. Multiplying by the nominal frequency to get the specification in Hertz, the target region is $[59.52,60.48]$ and the avoid region is $(-\infty,59.1)$. \Cref{fig:noInput} shows that the specification is violated when no control is applied.
\begin{figure}
    \centering
    \includegraphics[width=0.9\columnwidth]{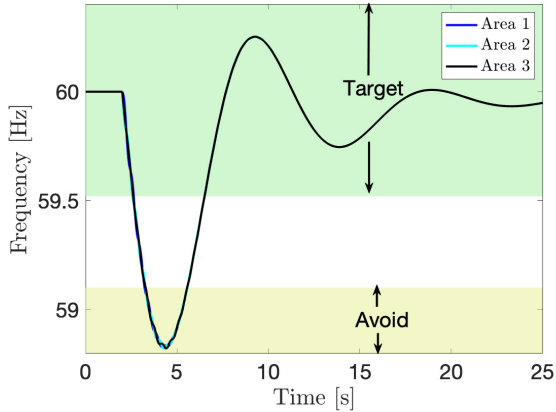}
    \caption{3A3M power system frequency without applying any control input.
    The frequency falls below $59.1$ Hz thus violates the specification.
    }
    \label{fig:noInput}
  \end{figure}

We apply the data-driven approaches of Section~\ref{sec:main_results} (fixed discretisation) and Section~\ref{sec:refinement} (abstraction refinement).
Both controllers are synthesised with disturbance $W =[-0.2, 0.2]\times [-0.3, 0.3]$. 
A comparison of the two control approaches is shown in Table~\ref{tab:3A3M}.
The required number of sample trajectories for each $(\xh,\uh)$ is computed using equation \eqref{eq:sample_complexity} and marked with $N$ in the table. 
The abstraction refinement starts with $\eta_x=0.012$ and refines the discretisation iteratively with a factor of two. The algorithm successfully finds a controller after five iterations.
The runtimes and the resulting winning region sizes $|\mathcal V|$ are also given in Table~\ref{tab:3A3M}. The abstraction refinement synthesises the controller a factor of $100$ times faster than the fixed discretisation by iteratively decreasing the value of $\eta_x$.


\begin{table}[h]
		
	\caption{Results for the 3A3M power system. 
	}
	\renewcommand{\arraystretch}{1.2}
	\setlength{\tabcolsep}{0.7em} 
	\resizebox{1\columnwidth}{!}{
	\begin{tabular}{l|cc|c|ccc}
		\toprule
		Control Approach&\multicolumn{2}{c|}{Dimension}& Disturbance& 
		\multicolumn{3}{c}{}\\
		
		&$X$&$U$&$\bar w$&$N$& time (min)&$|\mathcal V|$\\
		\midrule
		\text{Fixed Discretisation} &\multirow{2}{*}{\rotatebox{0}{$3$}}&\multirow{2}{*}{\rotatebox{0}{$1$}}&$(0.2,0.3)$& $3,290$&$5,253$  & $230,760$\\
		\text{Adaptive Refinement}&&&$(0.2,0.3)$& $4,460$&$50.25$  & $314,802$\\
		
		
		\bottomrule
	\end{tabular}}
	\label{tab:3A3M}
\end{table}

The data-driven control approach with fixed discretisation is simulated in PST and is reported in \Cref{fig:frequency,fig:input}.
The controlled system successfully keeps the frequencies of the three areas outside of the avoid set (i.e., always above $59.1$ Hz) and bring them back to the target set (i.e., above $59.52$ Hz).
\Cref{fig:input} shows the load changes in the system. Load at bus $1$ is able to maintain the frequencies of the three areas above the avoid region and facilitate the system returning to the target set for the maximum disturbances applied at buses $2,3$.
\Cref{fig:frequency_ad,fig:input_ad} show the results of simulating the system in PST with the control obtained from the abstraction refinement approach. The controlled system has the same performance in satisfying the specification.

\begin{figure}
    \centering
    \includegraphics[width=0.9\columnwidth]{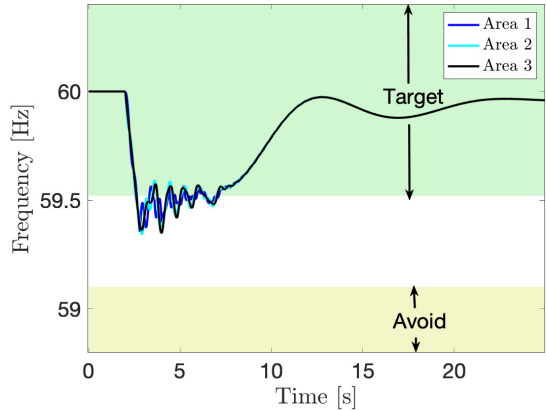}
    \caption{3A3M power system frequencies for the three areas, with the frequency of an area is measured at the corresponding bus in that area. The control synthesised by the fixed discretisation approach successfully keeps the frequencies of the three areas outside of the avoid set. The frequencies leave the target set for around $4.4$ seconds before staying in the target set.
    }
    \label{fig:frequency}
  \end{figure}

\begin{figure}
  \centering
    \centering
    \includegraphics[width=0.9\columnwidth]{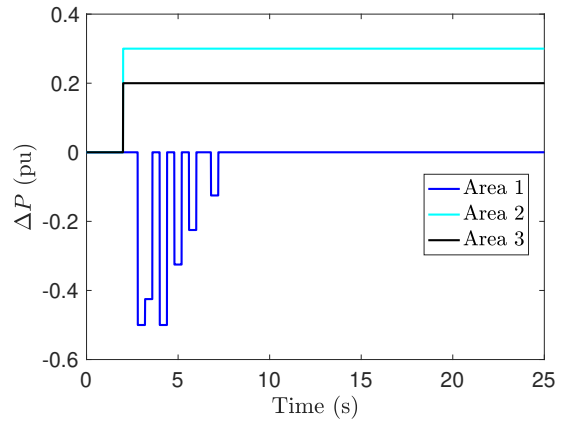}
    \caption{3A3M power system load changes for the three areas. Loads at buses $2$ and $3$ increase by $0.3$ and $0.2$ pu, respectively. 
    Load at bus $1$ is used to control the frequency using our data-driven approach with fixed discretisation.}
    \label{fig:input}
  \end{figure}

\begin{figure}
    \centering
    \includegraphics[width=0.9\columnwidth]{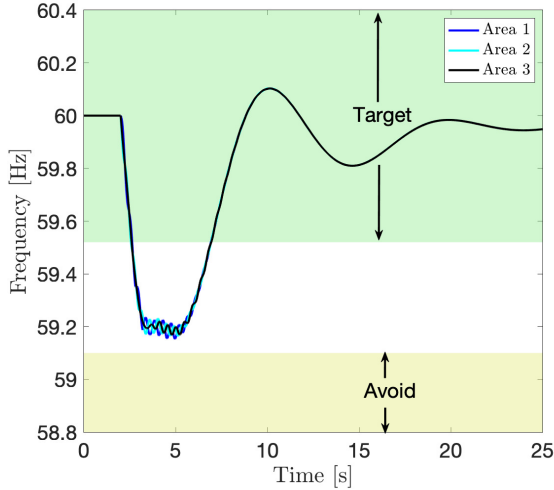}
    \caption{
    3A3M power system frequencies for the three areas, with the frequency of an area is measured at the corresponding bus in that area. The control synthesised by the abstraction refinement approach successfully satisfies the specification. The frequencies leave the target set for around $4.2$ seconds before staying in the target set.
   }
    \label{fig:frequency_ad}
  \end{figure}

\begin{figure}
    \centering
    \includegraphics[width=0.9\columnwidth]{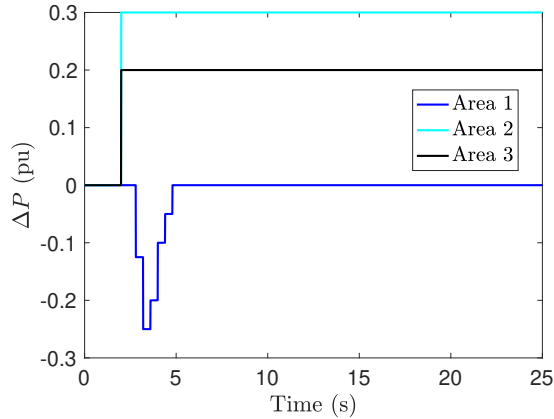}
    \caption{
    3A3M power system load changes for the three areas. Loads at buses $2$ and $3$ increase by $0.3$ and $0.2$ pu, respectively. 
    Load at bus $1$ is used to control the frequency using our data-driven approach with abstraction refinement.
    }
    \label{fig:input_ad}
  \end{figure}
  
  \subsection{Comparison with PAC Learning}
  In this subsection, we compare our approach with the results provided by Xue et al.~\cite{xue2020pac} that is based on probably approximately correct (PAC) learning on the 3A3M power system case study.
 The abstraction approach of \cite{xue2020pac} has no bias term $\gamma$,  but uses confidence parameter $\beta\in(0,1)$, error level $\epsilon\in(0,1)$, and cardinality of the parameter vector $\theta$ denoted by $q\in\nats$.
The required number of samples is
\begin{equation}
	N \geq \frac{2}{\epsilon}(ln \frac{1}{\beta} + q),
\end{equation}
which allows the constructed abstraction to hold for the entire state space except a subset measured by parameter $\epsilon$.

\begin{table}
	\caption{Comparing the winning domain of controllers obtained from our RSA method, PAC method of \cite{xue2020pac}, and the model-based approach of \cite{reissig2016feedback}.
	The pairwise comparison is made by computing the intersections ($\cap$) and set differences ($\text{row}\setminus\text{column}$). The results are reported both in cardinalities and percentages.}
	\renewcommand{\arraystretch}{1.2}
	\setlength{\tabcolsep}{0.7em}
	\resizebox{1\columnwidth}{!}{
	\begin{tabular}{l|cc|cc|cc}
		\toprule
		Winning Domain &\multicolumn{2}{c|}{RSA}& \multicolumn{2}{c|}{PAC}& 
		\multicolumn{2}{c}{Model-based}\\
		
		&$\cap$&$\backslash$&$\cap$&$\backslash$&$\cap$&$\backslash$\\
		\midrule
		RSA &$230,760$&$0$&$230,760$&$0$&$230,760$&$0$\\
		$\%$	&$100.00\%$&$0.00\%$&$100.00\%$&$0.00\%$&$100.00\%$&$0.00\%$\\
		
		\midrule
		PAC &$230,760$&$15,664$&$246,424$&$0$&$245,345$&$1,079$\\
		$\%$	&$93.64\%$&$6.36\%$&$100.00\%$&$0.00\%$&$99.56\%$&$0.44\%$\\
		
		\midrule
		Model-based &$230,760$&$22,216$&$245,345$&$7,631$&$252,976$&$0$\\
		$\%$&$91.22\%$&$8.78\%$&$96.98\%$&$3.02\%$&$100.00\%$&$0.00\%$\\
		
		\bottomrule
	\end{tabular}}
	
	\label{tab:compare}
\end{table}

We implement our data-driven robust scenario approach (RSA), the PAC approach in \cite{xue2020pac} with parameters $\beta = 0.01$ and $\epsilon = 0.01$, and the model-based approach of \cite{reissig2016feedback}.
Table~\ref{tab:compare} compares the winning domain of the controllers by reporting the intersections ($\cap$) and set differences ($\text{row}\setminus\text{column}$). It can be seen that the winning domain obtained by our RSA method 
 is a subset of the ones computed by PAC and the model-based approaches.
 This shows that our approach is more conservative than the model-based approach but correctly finds a subset of the winning domain. In contract, the PAC approach gives a winning domain that includes states not identified winning by the model-based approach. It includes $1079$ states outside of the winning domain obtained by the model-based approach.
Due to the nature of the PAC learning, some of these states are incorrectly identified as winning. The main reason is that the PAC method may miss to capture some of the transitions and does not always generate an overapproximation of the system behaviours. 
Among these $1079$ states, a counter example can be found, demonstrating a lack of guarantee provided by the PAC method. At state $(0.0187,0.0262 ,-0.1163 )$ the PAC controller calculates $u = -0.075$ to be an input which will transition to a safe state under any disturbances. However, the system under disturbances $W_1 = 0.2$ and $W_2 = 0.3$ will lead to the state $(0.0188  , 0.0131 , -0.1167)$ that is outside of the winning domain of the controller. 
In comparison, the winning domain provided by our RSA method is a subset of the one from the model-based method and provides full guarantees on the satisfaction of the specification and correctness of the controller. This guarantee is obtained at the cost of increased number of samples and a bias term included in the growth bound calculations, which makes the controller more conservative.

As a final point on this case study, note that our sampling approach uses the Lipschitz constant estimated using sample trajectories. This Lipschitz constant can in turn be used to construct the abstraction.
The direct use of the estimated Lipschitz constant does not provide a formal guarantee as it is an estimated value that converges to the true value only in the limit (i.e., the number of samples goes to infinity), and is likely to provide an overly conservative controller. On this particular case study, the direct use of the Lipschitz constant gives a controller that covers only $78.8\%$ of the winning domain of the model-based approach.

\subsection{Parameter Optimisation}

In this subsection, we discuss how selection of different parameters can affect the sample complexity and conservativeness of our method. We fix the path planning case study with the estimated Lipschitz constant $1.46$.
\Cref{fig:eps_N,fig:beta_N} illustrate the effect of changing parameters $\varepsilon,\beta$ on the number of samples $N$ required for each pair $(\xh,\uh)$ in order to compute the growth bound with confidence $(1-\beta)$.
\Cref{fig:beta_N} illustrates the effect of increasing the confidence parameter $\beta$ on reducing the sample complexity, for a fixed $\varepsilon=0.01$.
\Cref{fig:eps_N} shows that for a fixed $\beta=0.01$, increasing $\varepsilon$ leads to a rapid drop in $N$. 
In both \Cref{fig:eps_N,fig:beta_N}, the sample complexity increases in the presence of disturbance as the dimension of the sample space becomes larger.

\Cref{fig:eps_bias} demonstrate the effect of changing $\varepsilon$ on the value of the bias term $\gamma$ that makes the inequalities of the SCP more conservative. The bias term $\gamma$ increases for larger values of $\varepsilon$. Therefore, increasing $\varepsilon$ can decrease the sample complexity while increasing $\gamma$. Finally, it can be observed that the value of $\gamma$ is larger in the presence of disturbance. 



\begin{figure}
	\includegraphics[width=0.9\columnwidth]{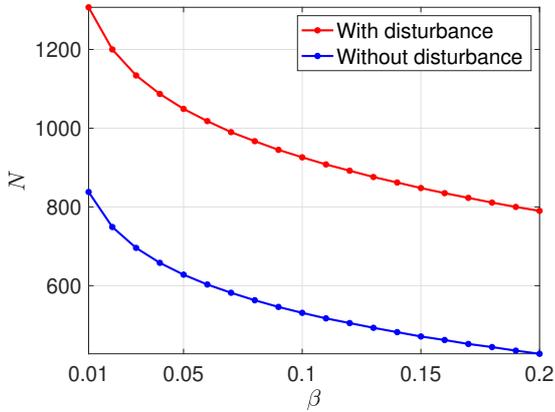}
	\caption{Required number of samples for our approach as a function of $\beta$ for a fixed $\varepsilon=0.01$.}
	\label{fig:beta_N}
\end{figure}
\begin{figure}
	\includegraphics[width=0.9\columnwidth]{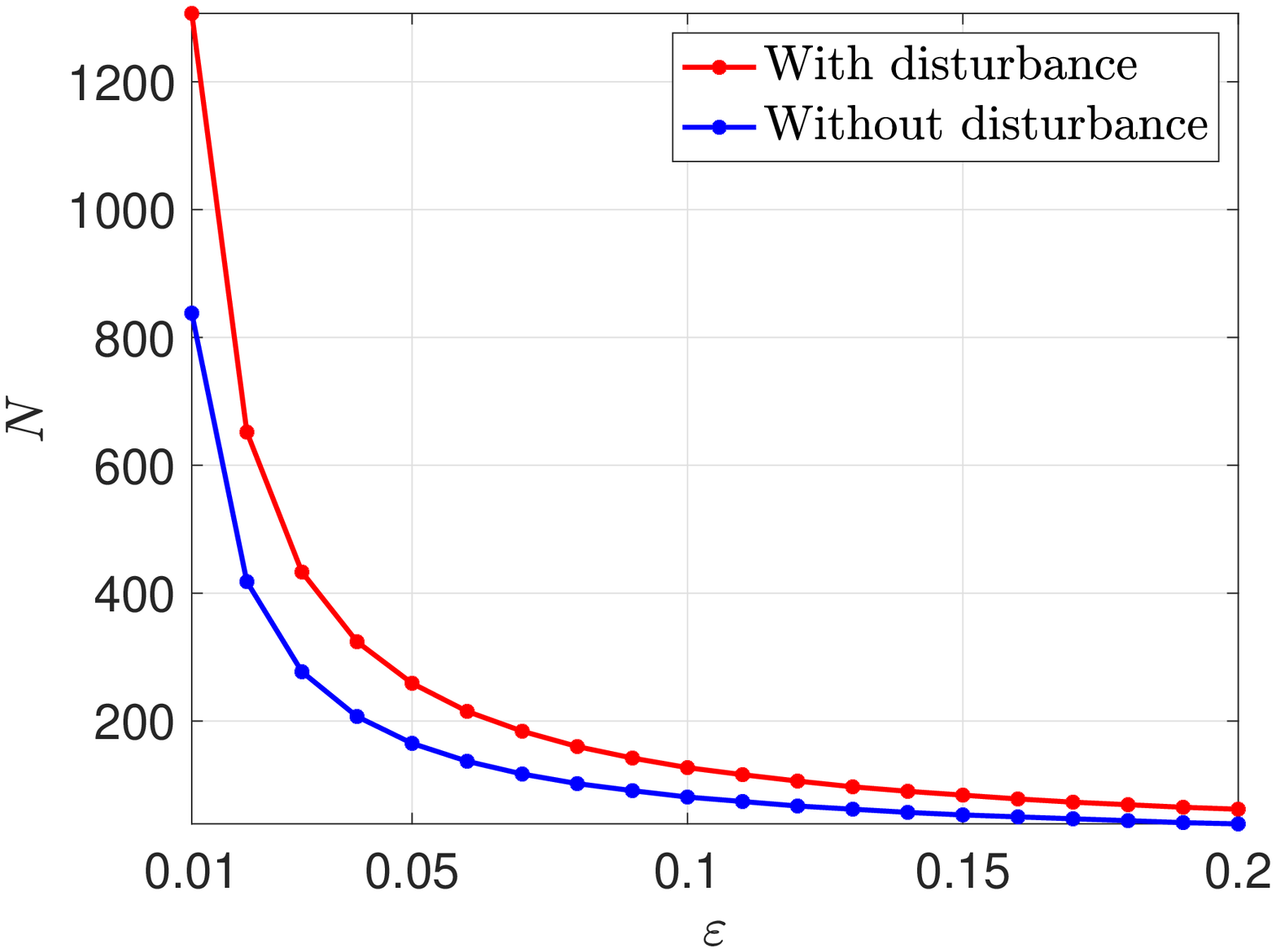}
	\caption{Required number of samples for our approach as a function of $\varepsilon$ for a fixed $\beta=0.01$. 
	}
	\label{fig:eps_N}
\end{figure}
\begin{figure}
	\includegraphics[width=0.9\columnwidth]{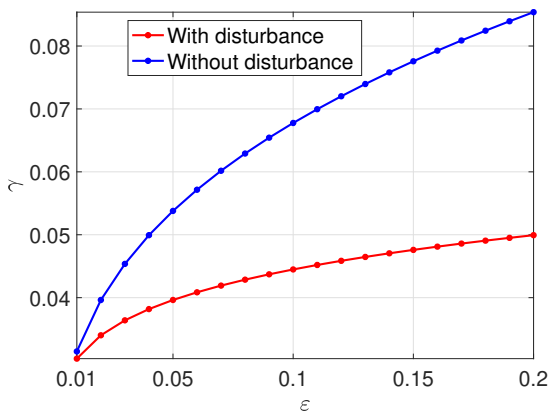}
	\caption{The bias term $\gamma$ as a function of $\varepsilon$.}
	\label{fig:eps_bias}
\end{figure}

\section{Discussion and Future Work}\label{sec:conclusion}
We proposed a data-driven method for computing finite abstractions of continuous systems with unknown dynamics. Our approach casts the computation of an overapproximation of reachable sets as a robust convex program (RCP). A feasible solution for the RCP is then obtained with a given confidence by solving a corresponding scenario convex program (SCP). The SCP does not need the dynamics of the system and requires only a finite set of sample trajectories. We provided a sample complexity result that gives a lower-bound on the number of trajectories to achieve a certain confidence.
Our sample complexity results requires knowing a bound on the Lipschitz constant of the system, that we estimated using extreme value theory.

We guaranteed that with high confidence, the computed abstraction is a valid abstraction of the system that overapproximates its behaviours on its entire state space.
We showed that our data-driven approach can be embedded into abstraction refinement schemes for designing a controller and enlarging the winning region of the controller with respect to satisfaction of temporal properties. Finally, we evaluated our approach on three case studies.

In the future, we plan to extend our approach by enlarging the class of disturbances beyond piece-wise constant ones (i.e., tackling the issue of infinite dimensional sampling spaces), improve scalability of the approach by providing more efficient parallel implementation of the approach, and apply it to large case studies that are combinations of differential equations, block diagrams, and lookup tables.


\bibliographystyle{plain}
\bibliography{references}

\end{document}